\documentclass{article}

\usepackage{arxiv}

\usepackage{cite}
\usepackage[utf8]{inputenc} % allow utf-8 input
\usepackage{authblk}
\usepackage{graphicx}
\usepackage{amssymb}
\usepackage{amsmath}
\usepackage{amsthm}
\usepackage{algorithm}
\usepackage[noend]{algpseudocode}
\usepackage{titletoc}
\usepackage{titlesec}
\usepackage[utf8]{inputenc}

\makeatletter
\renewcommand{\Function}[2]{%
	\csname ALG@cmd@\ALG@L @Function\endcsname{#1}{#2}%
	\def\jayden@currentfunction{#1}%
}
\newcommand{\funclabel}[1]{%
	\@bsphack
	\protected@write\@auxout{}{%
		\string\newlabel{#1}{{\jayden@currentfunction}{\thepage}}%
	}%
	\@esphack
}
\makeatother

\newtheorem{thm}{Theorem}[section]
\newtheorem{lem}[thm]{Lemma}
\newtheorem{prop}[thm]{Proposition}

\newtheorem{defn}{Definition}[section]

\usepackage[english]{babel}

\makeatletter
\def\BState{\State\hskip-\ALG@thistlm}
\makeatother

\title{OPERA: Reasoning about continuous common knowledge in asynchronous distributed systems}

\author[1]{\large Sang-Min Choi}
\author[1]{Jiho Park}
\author[1]{Quan Nguyen}
\author[1]{Andre Cronje}
\author[2]{Kiyoung Jang}
\author[2]{Hyunjoon Cheon}
\author[2]{Yo-Sub Han}
\author[1]{Byung-Ik Ahn}

\affil[1]{FANTOM Lab\\ FANTOM Foundation}
\affil[2]{Department of Computer Science\\ Yonsei University}

\begin{document}
\maketitle

\begin{abstract}
	This paper introduces a new family of consensus protocols, namely \emph{Lachesis-class} denoted by $\mathcal{L}$, for distributed networks with guaranteed Byzantine fault tolerance. Each Lachesis protocol $L$ in $\mathcal{L}$ has complete asynchrony, is leaderless, has no round robin, no proof-of-work, and has eventual consensus.
	
	The core concept of our technology is the \emph{OPERA chain}, generated by the Lachesis protocol. In the most general form, each node in Lachesis has a set of $k$ neighbours of most preference. When receiving transactions a node creates and shares an event block with all neighbours. Each event block is signed by the hashes of the creating node and its $k$ peers. The OPERA chain of the event blocks is a Directed Acyclic Graph (DAG); it guarantees practical Byzantine fault tolerance (pBFT). Our framework is then presented using Lamport timestamps and concurrent common knowledge.
	
	Further, we present an example of Lachesis consensus protocol $L_0$ of our framework. Our $L_0$ protocol can reach consensus upon 2/3 of all participants' agreement to an event block without any additional communication overhead. $L_0$ protocol relies on a cost function to identify $k$ peers and to generate the DAG-based OPERA chain. By creating a binary flag table that stores connection information and share information between blocks, Lachesis achieves consensus in fewer steps than pBFT protocol for consensus.
\end{abstract}


% Table of Contents
\newpage
\pagenumbering{arabic} 
\tableofcontents 
% \linenumbers
\newpage
\section{Introduction}\label{ch:intro}

Beyond the success of cryptocurrencies, blockchain has recently emerged as a technology platform that offers secure decentralized consistent transaction ledgers and has powered innovations across domains including financial systems, supply chains and health care.
Despite the high demand in distributed ledger technology~\cite{bcbook15}, commercialization opportunities have been obstructed by long processing time for consensus, and high power consumption. These issues have been addressed in consensus algorithms such as \cite{algorand16, algorand17, sompolinsky2016spectre, PHANTOM08}. 

Distributed database systems often address \emph{Byzantine} fault tolerance~\cite{Lamport82} in which up to just under one-third of the participant nodes may be compromised. Consensus algorithms ensures the integrity of transactions between participants over a distributed network~\cite{Lamport82} and is equivalent to the proof of \emph{Byzantine} fault tolerance in distributed database systems~\cite{randomized03, paxos01}. 
Byzantine consensus is not guaranteed for deterministic, completely asynchronous system with unbounded delays~\cite{flp}. But achieving consensus is feasible for nondeterministic system with probability one. 

There are several approaches to consensus in distributed system.
The original Nakamoto consensus protocol in Bitcoin uses Proof of Work (PoW), which requires large amounts of computational work to generate the blocks by participants~\cite{bitcoin08}.  Alternative schemes such as Proof Of Stake (PoS)~\cite{ppcoin12,dpos14} have been proposed. PoS uses participants' stakes to generate the blocks respectively. 
Another approach utilizes directed acyclic graphs (DAG)~\cite{dagcoin15, sompolinsky2016spectre, PHANTOM08, PARSEC18, conflux18} to facilitate consensus. 

Examples of DAG-based consensus algorithms include Tangle~\cite{tangle17}, Byteball~\cite{byteball16}, and Hashgraph~\cite{hashgraph16}. Tangle selects the blocks to connect in the network utilizing accumulated weight of nonce and Monte Carlo Markov Chain (MCMC). Byteball generates a main chain from the DAG and reaches consensus through index information of the chain. Hashgraph connects each block from a node to another random node. Hashgraph searches whether 2/3 members can reach each block and provides a proof of Byzantine fault tolerance via graph search.

\subsection{Motivation}

Practical Byzantine Fault Tolerance (pBFT) allows all nodes to successfully reach an agreement for a block (information) when a Byzantine node exists \cite{Castro99}. In pBFT, consensus is reached once a created block is shared with other participants and the share information is shared with others again \cite{zyzzyva07, honey16}. After consensus is achieved, the block is added to the participants’ chains~\cite{Castro99, Blockmania18}.
Currently, it takes $O(N^4)$ for pBFT.

HashGraph~\cite{hashgraph16} proposes ``gossip about gossip" and virtual voting to reach consensus. There are several limitations with HashGraph.
First, the algorithm operates on a known network, which needs full awareness of all authoritative participants. Second, gossip propagation is slow and latency increases to $O(n)$ with $n$ participants. Third, it remains unclear whether virtual voting is faster than chain weight aka longest chain/proof of work concept. These issues are gossip problems and not consensus problems.

We are interested in a new approach to address the aforementioned issues in pBFT approaches \cite{Castro99,zyzzyva07, honey16} and HashGraph~\cite{hashgraph16}. Specifically, we propose a new consensus algorithm that addresses the following questions: (1) Can we reach local consensus in a $k$-cluster faster for some $k$?, (2) Can we make gossips faster such as using a broadcast based gossip subset?, (3) Can continuous common knowledge be used for consensus decisions with high probability? (4) Can complex decisions be reduced to binary value consensus?

In this paper, we propose a new approach that can quickly search for Byzantine nodes within the block DAG.
In particular, we introduce a new class of consensus protocols, namely Lachesis protocol denoted by $\mathcal{L}$. The core idea of Lachesis is to use a new DAG structure, the OPERA chain, which allows faster path search for consensus.
We then propose an example of the Lachesis protocol class, which is called the Lachesis protocol $L_0$.

\subsection{Generic framework of $\mathcal{L}$ Protocols}
We introduce a generic framework of Lachesis protocols, called $\mathcal{L}$. The basic idea of Lachesis protocol is a DAG-based asynchronous non-deterministic protocol that guarantees pBFT.
We propose OPERA chain --- a new DAG structure for faster consensus. Lachesis protocol generates each block asynchronously and the Lachesis algorithm achieves consensus by confirming how many nodes know the blocks using the OPERA chain.
Figure~\ref{fig:operachain} shows an example of OPERA chain constructed through a Lachesis protocol.
\begin{figure}[h] \centering
\includegraphics[height=7cm, width=1.0\columnwidth]{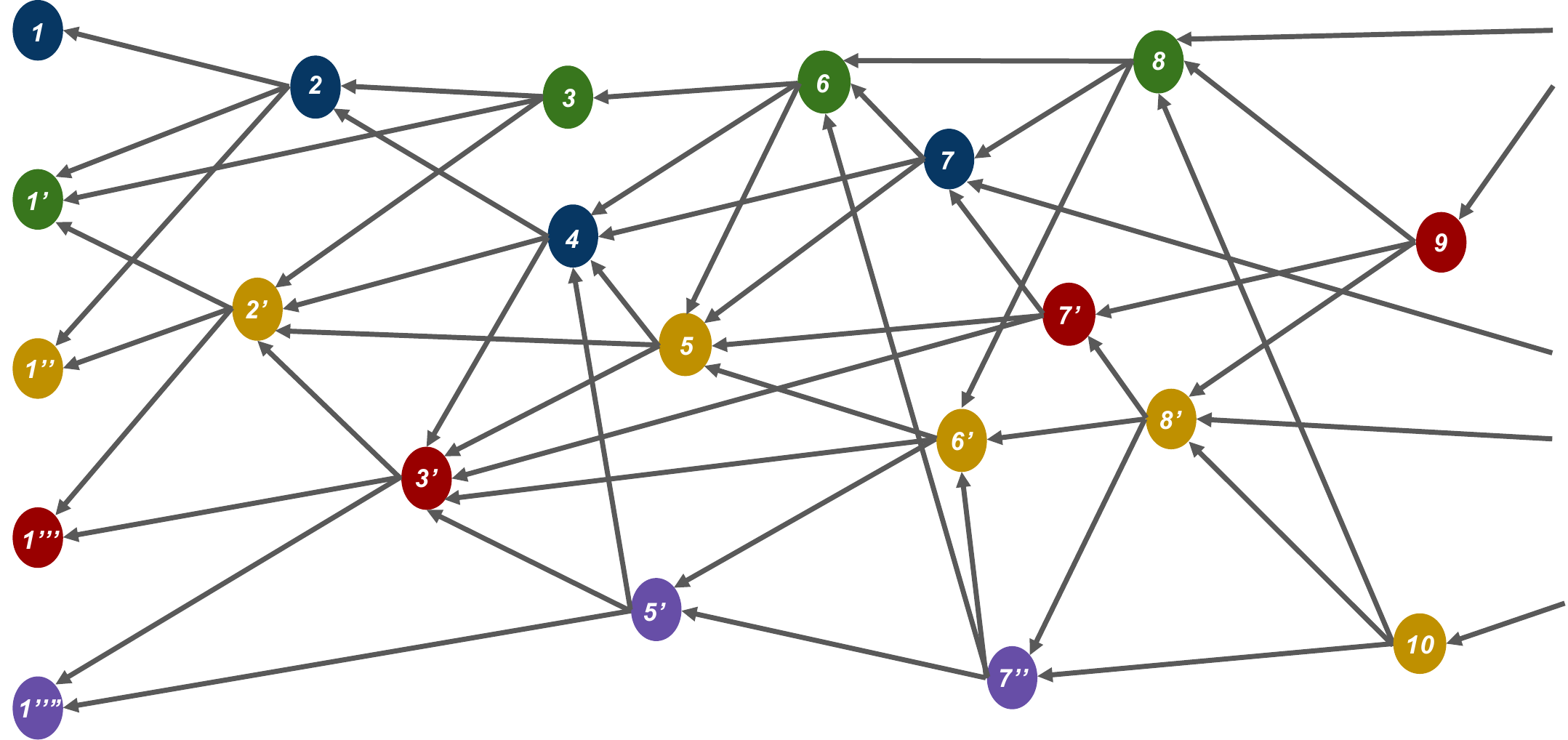}
\caption{An Example of OPERA Chain}
\label{fig:operachain}
\end{figure}

The main concepts of Lachesis are given as follows:
\begin{description}
\item[$\bullet$ Event block] All nodes can create event blocks as time $t$. The structure of an event block includes the signature, generation time, transaction history, and hash information to references. The information of the referenced event blocks can be copied by each node. The first event block of each node is called a \emph{leaf event}.
\item[$\bullet$ Lachesis protocol] Lachesis protocol is the rule-set to communicate between nodes. When each node creates event blocks, it determines which nodes choose other nodes to broadcast to. Node selection can be random or via some cost function.
\item[$\bullet$ Happened-before] Happened-before is the relationship between nodes which have event blocks. If there is a path from an event block $x$ to $y$, then $x$ Happened-before $y$. ``$x$ Happened-before $y$" means that the node creating $y$ knows event block $x$.
\item[$\bullet$ Root] An event block is called a \emph{root} if either (1) it is the first generated event block of a node, or (2) it can reach more than two-thirds of other roots. Every root can be candidate for Clotho.
\item[$\bullet$ Root set] Root set ($R_s$) is the set of all roots in the frame. The cardinality of the set is $2n/3 < R_s \leq$ $n$, where $n$ is the number of all nodes.  
\item[$\bullet$ Frame] Frame $f$ is a natural number that separates Root sets. The frame increases by 1 in case of a root in the new set ($f+1$). And all event blocks between the new set and the previous Root set are included in the frame $f$. 
\item[$\bullet$ Flag table] The Flag table stores reachability from an event block to another root. The sum of all reachabilities, namely all values in flag table, indicates the number of reacheabilities from an event block to other roots. 
\item[$\bullet$ Lamport timestamps] For topological ordering, Lamport timestamps algorithm uses the happened-before relation to determine a partial order of the whole event block based on logical clocks.
\item[$\bullet$ Clotho] A Clotho is a root that satisfies that they are known by more than $2n/3$ nodes and more than $2n/3$ nodes know the information that they are known in nodes. A Clotho can be a candidate for Atropos. 
\item[$\bullet$ Atropos] An Atropos is assigned consensus time through the Lachesis consensus algorithm and is utilized for determining the order between event blocks. Atropos blocks form a Main-chain, which allows time consensus ordering and responses to attacks.
\item[$\bullet$ Reselection] To solve the byzantine agreement problem, each node reselects a consensus time for a Clotho, based on the collected consensus time in the root set of the previous frame. When the consensus time reaches byzantine agreement, a Clotho is confirmed as an Atropos and is then used for time consensus ordering.
\item[$\bullet$ OPERA chain] The OPERA chain is the local view of the DAG held by each node, this local view is used to identify topological ordering, select Clotho, and create time consensus through Atropos selection.
\item[$\bullet$ Main-Chain] Main-chain is a core subset of the OPERA chain. It is comprised of Atropos event blocks. Thus, the OPERA chain uses Main-chain to find rapid ordering between event blocks. In OPERA chain, each event block is assigned a proper consensus position.
\end{description}

\begin{figure}[h] \centering
\includegraphics[height=7cm, width=1.0\columnwidth]{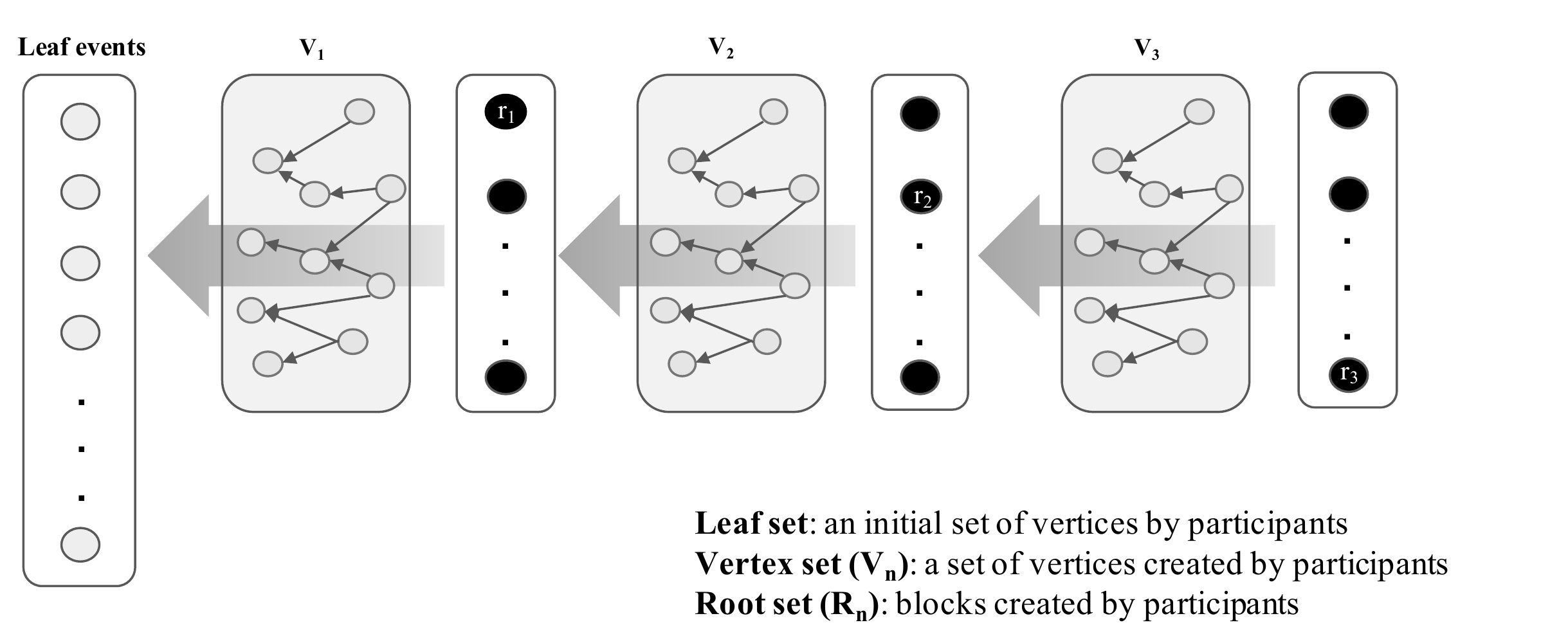}
\caption{Consensus Method through Path Search in a DAG (combines chain with consensus process of pBFT)}
\label{fig:pBFTtoPath}
\end{figure}
As a motivating example, Figure~\ref{fig:pBFTtoPath} illustrates how consensus is reached through the path search in the OPERA chain. In the figure, leaf set, denoted by $R_{s0}$, consists of the first event blocks created by individual participant nodes. $V$ is the set of event blocks that do not belong neither in $R_{s0}$ nor in any root set $R_{si}$.
Given a vertex $v$ in $V \cup R_{si}$, there exists a path from $v$ that can reach a leaf vertex $u$ in $R_{s0}$. 
Let $r_1$ and $r_2$ be root event blocks in root set $R_{s1}$ and $R_{s2}$, respectively.
$r_1$ is the block where a quorum or more blocks exist on a path that reaches a leaf event block. 
Every path from $r_1$ to a leaf vertex will contain a vertex in $V_1$. Thus, if there exists a vertex $r$ in $V_1$ such that $r$ is created by more than a quorum of participants, then $r$ is already included in $R_{s1}$. Likewise, $r_2$ is a block that can be reached for $R_{s1}$ including $r_1$ through blocks made by a quorum of participants.
For all leaf event blocks that could be reached by $r_1$, they are shared with more than quorum participants through the presence of $r_1$. The existence of the root $r_2$ shows that information of $r_1$ is shared with more than a quorum. 
This kind of a path search allows the chain to reach consensus in a similar manner as the pBFT consensus processes. It is essential to keep track of the blocks satisfying the pBFT consensus process for quicker path search; our OPERA chain and Main-chain keep track of these blocks.

\subsection{Lachesis protocol $L_0$}
We now introduce a new specific Lachesis consensus protocol, called $L_0$. The new protocol $L_0$ is a DAG-based asynchronous non-deterministic protocol that guarantees pBFT.
$L_0$ generates each block asynchronously and uses the OPERA chain for faster consensus by checking how many nodes know the blocks.

In this $L_0$ protocol, we propose several algorithms. In particular, we introduce an algorithm in which a node can identify lazy participants from cost-effective peers --- say its $k$ peers. We must stress that a generic Lachesis protocol does not depend on any $k$ peer selection algorithm; each node can choose $k$ peers randomly. Each message created by a node is then signed by the creating node and its $k$ peers. We also introduce a flag table data structure that stores connection information of event blocks. The flag table allows us to quickly traverse the OPERA chain to find reachability between event blocks.

OPERA chain can be used to optimize path search. 
By using certain event blocks (Root, Clotho, and Atropos), Main chain --- a core subgraph of OPERA chain, can maintain reliable information between event blocks and reach consensus.  Generating event blocks via Lachesis protocol, the OPERA chain and Main chain are updated frequently and can respond strongly to attack situations such as forking and parasite attack. Further, using the flag table over the OPERA chain, consensus can be quickly reached, and the ordering between specific event block can be determined.

\subsection{Contributions}

In summary, this paper makes the following contributions.

\begin{itemize}
\item We propose a new family $\mathcal{L}$ of Lachesis protocols. We introduce the OPERA chain and Main-chain for faster consensus.
\item We define a topological ordering of nodes and event blocks in the OPERA chain. By using Lamport timestamps, the ordering is more intuitive and reliable in distributed system. We introduce a flag table at each block to improve root detection.
\item We present proof of how a DAG-based protocol can implement concurrent common knowledge for consistent cuts.
\item The Lachesis protocols allow for faster node synchronization with $k$-neighbor broadcasts.

\item A specific Lachesis protocol $L_0$ is then introduced with specific algorithms. The benefits of Lachesis protocol $L_0$ include (1) root selection algorithm via flag table;  (2) an algorithm to build the Main-chain; (3) an algorithm for $k$ peers selection via cost function; (4) faster consensus selection via $k$ peer broadcasts; (5) data pruning via root creation. 

\end{itemize}

The rest of this paper is organised as follows. Section~\ref{se:Previous} gives an overview of Blockchain related work as well as existing DAG-based protocols.        
Section~\ref{se:protocol} describes our new Lachesis protocol. Section~\ref{se:lca} presents Lachesis consensus algorithm.
Several discussions about Lachesis protocols are presented in Section~\ref{se:discuss}. Section~\ref{se:con} concludes with some future work. Section~\ref{se:appendix}. Proof of Byzantine fault tolerance is described in Section~\ref{se:proof}. In Section~\ref{se:ra}, we present  responses to certain attacks with the Lachesis protocol and consensus algorithm.

\section{Related work}\label{se:Previous}

\subsection{Lamport timestamps}

Lamport~\cite{lamport1978time} defines the "happened before" relation between any pair of events in a distributed system of machines. The happened before relation, denoted by $\rightarrow$, is defined without using physical clocks to give a partial ordering of events in the system. The relation "$\rightarrow$" satisfies the following three conditions: (1) If $b$ and $b'$ are events in the same process, and $b$ comes before $b'$, then $b \rightarrow b'$. (2) If $b$ is the sending of a message by one process and $b'$ is the receipt of the same message by another process, then $b \rightarrow b'$. (3) If $b \rightarrow b'$ and $b' \rightarrow b''$ then $b \rightarrow b''$. 
Two distinct events $b$ and $b'$ are said to be concurrent if $b \nrightarrow b'$ and $b' \nrightarrow b$.

The happens before relation can be viewed as a causality effect: that $b \rightarrow  b'$ implies event $b$ may causally affect event $b'$. Two events are concurrent if neither can causally affect the other.

Lamport introduces logical clocks which is a way of assigning a number to an event. A clock $C_i$ for each process $P_i$ is a function which assigns a number $C_i(b)$ to any event $b \in P_i$. The entire system of blocks is represented by the function $C$ which assigns to any event $b$ the number $C(b)$, where $C(b) = C_j(b)$ if $b$ is an event in process $P_j$.
The Clock Condition states that for any events $b$, $b'$: if $b \rightarrow b'$ then $C(b)$ $<$ $C(b')$.

To satisfies the Clock Condition, the clocks must satisfy two conditions. First, each process $P_i$ increments $C_i$ between any two successive events. Second, we require that each message $m$ contains a timestamp $T_m$, which equals the time at which the message was sent. Upon receiving a message timestamped $T_m$, a process must advance its clock to be later than $T_m$.

Given any arbitrary total ordering $\prec$ of the processes, the total ordering $\Rightarrow$ is defined as follows: if $a$ is an event in process $P_i$ and $b$ is an event in process $P_j$, then $b \Rightarrow b'$ if and only if either (i) $C_i(b) < C_j(b')$ or (ii) $C(b)= Cj(b')$ and $P_i \prec P_j$. The Clock Condition implies that if $b \rightarrow b'$ then $b \Rightarrow b'$.

\subsection{Concurrent common knowledge}\label{se:cck}

In the Concurrent common knowledge (CCK) paper ~\cite{cck92}, they define a model to reason about the concurrent common knowledge in asynchronous, distributed systems. A system is composed of a set of processes that can communicate only by sending messages along a fixed set of channels. The network is not necessarily completely connected. The system is asynchronous in the sense that there is no global clock in the system, the relative speeds of processes are independent, and the delivery time of messages is finite but unbounded.

A local state of a process is denoted by $s^j_i$. Actions are state transformers; an action is a function from local states to local states. An action can be either: a send(m) action where m is a message, a receive(m) action, and an internal action.
A local history, $h_i$, of process $i$, is a (possibly infinite) sequence of alternating local states—beginning with a distinguished initial state—and actions. We write such a sequence as follows: 
$h_i = s_i^0 \xrightarrow{ \alpha_i^1 } s_i^1 \xrightarrow{\alpha_i^2} s_i^2 \xrightarrow{\alpha_i^3} ...$
The notation of $s^j_i$ ($\alpha^j_i$) refers to the $j$-th state (action) in process $i$'s local history
An event is a tuple $\langle s , \alpha, s' \rangle$ consisting of a state, an action, and a state.
The $j$th event in process $i$'s history is $e^j_i$ denoting $\langle s^{j-1}_i , \alpha^j_i, s^j_{i} \rangle$.

An asynchronous system consists of the following sets.
\begin{enumerate}
\item A set $P$ = \{1,...,$N$\} of process identifiers, where $N$ is the total number of processes in the system.
\item $A$ set $C$ $\subseteq$ \{($i$,$j$) s.t. $i,j \in P$\} of channels. The occurrence of $(i,j)$ in $C$ indicates that process $i$ can send messages to process $j$.
\item A set $H_i$ of possible local histories for each process $i$ in $P$.
\item A set $A$ of asynchronous runs. Each asynchronous run is a vector of local histories, one per process, indexed by process identifiers. Thus, we use the notation
$a = \langle h_1,h_2,h_3,...h_N \rangle$.
Constraints on the set $A$ are described throughout this section.
\item A set $M$ of messages. A message is a triple $\langle i,j,B \rangle$ where $i \in P$ is the sender of the message, $j \in P$ is the message recipient, and $B$ is the body of the message. $B$ can be either a special value (e.g. a tag to denote a special-purpose message), or some proposition about the run (e.g. “$i$ has reset variable $X$ to zero”), or both. We assume, for ease of exposition only, that messages are unique.
\end{enumerate}

The set of channels $C$ and our assumptions about their behavior induce two constraints on the runs in $A$. First, $i$ cannot send a message to $j$ unless $(i,j)$ is a channel. Second, if the reception of a message $m$ is in the run, then the sending of $m$ must also be in that run; this implies that the network cannot introduce spurious messages or alter messages.

The CCK model of an asynchronous system does not mention time. Events are ordered based on Lamport's happens-before relation. They use Lamport’s theory to describe global states of an asynchronous system. A global state of run $a$ is an $n$-vector of prefixes of local histories of $a$, one prefix per process.
The happens-before relation can be used to define a consistent global state, often termed a consistent cut, as follows. 

\begin{defn}[Consistent cut] A consistent cut of a run is any global state such that if $e^x_i \rightarrow e^y_j$ and $e^y_j$ is in the global state, then $e^x_i$ is also in the global state.
\end{defn} 

A message chain of an asynchronous run is a sequence of messages $m_1$, $m_2$, $m_3$, $\dots$, such that, for all $i$, $receive(m_i)$ $\rightarrow$  $send(m_{i+1})$. Consequently,
$send(m_1)$ $\rightarrow$ $receive(m_1)$ $\rightarrow$ $send(m_2)$ $\rightarrow$ $receive(m_2)$ $\rightarrow$ $send(m_3)$ $\dots$.

\subsection{Consensus algorithms}

In a consensus algorithm, all participant nodes of a distributed network share transactions and agree integrity of the shared transactions~\cite{Lamport82}. It is equivalent to the proof of Byzantine fault tolerance in distributed database systems~\cite{randomized03, paxos01}. 
The Practical Byzantine Fault Tolerance (pBFT) allows all nodes to successfully reach an agreement for a block when a Byzantine node exists \cite{Castro99}.

There are numerous consensus algorithms being proposed~\cite{algorand16, algorand17}. Proof of Work (PoW) requires
large amounts of computational work to generate the blocks~\cite{bitcoin08}.  Proof of Stake (PoS)~\cite{ppcoin12,dpos14} use participants' stakes and delegated participants' stake to generate the blocks respectively. Alternative schemes are proposed to improve algorithms using directed acyclic graphs (DAG)~\cite{dagcoin15}. These DAG-based approaches utilize the graph structures to decide consensus; blocks and connections are considered as vertices and edges, respectively.

\subsection{DAG-based Approaches}

IOTA~\cite{tangle17} published a DAG-based technology called Tangle. The Tips concept was used to address scalability issues with the limitations of the Internet of Things. Also, a nonce by using weight level was composed to achieve the transaction consensus by setting the user’s difficulty. To solve the double spending problem and parasite attack, they used the Markov Chain Monte Carlo (MCMC) tip selection algorithm, which randomly selects the tips based on the size of the accumulated transaction weights. However, if a transaction conflicts with another, there is still a need to examine all past transaction history to find the conflict.

Byteball~\cite{byteball16} uses an internal pay system called bytes. This is used to pay for adding data to the distributed database. Each storage unit is linked to each other that includes one or more hashes of earlier storage units. In particular, the consensus ordering is composed by selecting a single Main Chain, which is determined as a root consisting of the most roots. 
A majority of roots detects the double-spend attempts through consensus time of Main Chain. The fee is charged according to the size of the bytes, and the list of all units should be searched and updated in the process of determining the roots.

RaiBlocks~\cite{raiblock17} has been developed to improve high fees and slow transaction processing.  It is a process of obtaining consensus through the balance weighted vote on conflicting transactions. Each node participating in the network becomes the principal and manages its data history locally. However, since RaiBlocks generate transactions in a similar way to an anti-spam tool of PoW, all nodes must communicate to create transactions. In terms of scalability, there is a need for steps to verify the entire history of transactions when a new node is added.

Hashgraph~\cite{hashgraph16} is an asynchronous DAG-based distributed ledger. Each node is connected by its own ancestor and randomly communicates known events through a gossip protocol. At this time, any famous node can be determined by the \textit{see} and \text{strong see} relationship at each round to reach consensus quickly. They state that if more than 2/3 of the nodes reach consensus for an event, it will be assigned consensus position.

Conflux~\cite{conflux18} is a DAG-based Nakamoto consensus protocol. Conflux is a fast, scalable and decentralized block chain system that optimistically processes concurrent blocks without discarding any as forks. The Conflux protocol achieves consensus on a total order of the blocks.
The total order of the transactions is decided by all participants of the network. Conflux can tolerate up to half of the network as malicious while the BFT-based approaches can only tolerate up to one third of malicious nodes.

Parsec~\cite{PARSEC18} proposes an algorithm for reaching consensus in the presence of Byzantine faults in a randomly synchronous network. 
Like Hashgraph~\cite{hashgraph16}, it has no leaders, no round robin, no proof-of-work and reaches eventual consensus with probability one. Unlike Hashgraph, it can provide high speed even in the presence of faults. Parsec algorithm reaches BFT consensus with very weak synchrony assumptions. Messages are delivered with random delays, such that the average delay is finite. It allows up to one-third Byzantine (arbitrary) failures.

Phantom~\cite{PHANTOM08} is a PoW based protocol for a permissionless ledger that generalizes Nakamoto’s blockchain to a DAG of blocks. PHANTOM includes a parameter $k$ to adjust the tolerance level of the protocol to blocks that were created concurrently, which can be set to accommodate higher throughput. It thus avoids the security-scalability tradeoff as in Satoshi’s protocol.  PHANTOM uses a greedy algorithm on the DAG to distinguish between blocks by honest nodes and those by non-cooperating nodes. This distinction gives PHANTOM a robust total order of the blocks that is eventually agreed upon by all honest nodes.

Similar to PHANTOM, the GHOSTDAG protocol selects a $k$-cluster, which induces a colouring of the blocks as Blues (blocks in the selected cluster) and Reds (blocks outside the cluster).
However, instead of searching for the largest $k$-cluster, GHOSTDAG finds a cluster using a greedy algorithm.

Spectre~\cite{sompolinsky2016spectre} is
a new protocol for the consensus core of cryptocurrencies. SPECTRE, which is PoW-based protocol, relies on a data structure that generalizes Nakamoto’s blockchain into a DAG. It remains secure from attackers with up to 50\% of the computational power even under high throughput and fast confirmation times. Sprectre protocol satisfies weaker properties than classic consensus requires. In SPECTRE, the order between any two transactions can be decided from transactions performed by honest users. This is different from the conventional paradigm in which the order must be decided by all non-corrupt nodes.

Blockmania~\cite{Blockmania18}
is a mechanism to achieve consensus with several advantages over the more traditional pBFT protocol and its variants. In Blockmania nodes in a quorum only emit blocks linking to other blocks, irrespective of the consensus state machine. The resulting directed acyclic graph of blocks (block DAG) is later interpreted to ensure consensus safety, finality and liveliness.
The resulting system has communication complexity $O(N^2)$ even in the worse case, and low constant factors — as compared to $O(N^4)$ for pBFT.

\section{Generic framework of Lachesis Protocols}\label{se:protocol}

This section describes the key concepts of our new family of Lachesis protocols.

\subsection{OPERA chain}

The core idea of Lachesis protocols is to use a DAG-based structure, called OPERA chain for our consensus algorithm. 
In Lachesis protocol, a (participant) node is a server (machine) of the distributed system.
Each node can create messages, send messages to and receive messages from other nodes. The communication between nodes is asynchronous.

Lachesis Protocol consists of event blocks including user information and edges between event blocks. In Lachesis Protocol, event blocks are created by a node after the node communicates information of OPERA chain with another node. The OPERA chain is comprised of event blocks as vertices and block communication as edges.

Let $n$ be the number of participant nodes.
For consensus, the algorithm examines whether an event block is \emph{shared} with $2n/3$ nodes, where $n$ is the number of all nodes. Sharing an event block with $2n/3$ nodes means that more than two-thirds of all nodes in the OPERA chain knows the event block. 

\subsection{Main-chain}

For faster consensus, we introduce the \emph{Main-chain}, which is a special sub-graph of the OPERA chain. To improve path search, we propose to use a local hash table structure as a cache that is used to quickly determine the closest root to an event block.

In the OPERA chain, an event block is called a \emph{root} if the event block is linked to more than two-thirds of previous roots. A leaf vertex is also a root itself. With root event blocks, we can keep track of ``vital'' blocks that $2n/3$ of the network agree on.

The Main chain --- a core subgraph of OPERA chain, plays the important role for ordering the event blocks. The Main chain stores shortcuts to connect between the Atropos. 
After the topological ordering is computed over all event blocks through Lachesis protocol, Atropos blocks are determined and form the Main chain.  
Figure~\ref{fig:mainchain} shows an example of Main chain composed of Atropos event blocks. 
In particular, the Main chain consists of Atropos blocks those are derived from root blocks and so are agreed by $2n/3$ of the network nodes. Thus, this guarantees that at least $2n/3$ of nodes have come to consensus on this Main chain. 

\begin{figure} [H] \centering  
	\includegraphics[height=8cm, width=1.0\columnwidth]{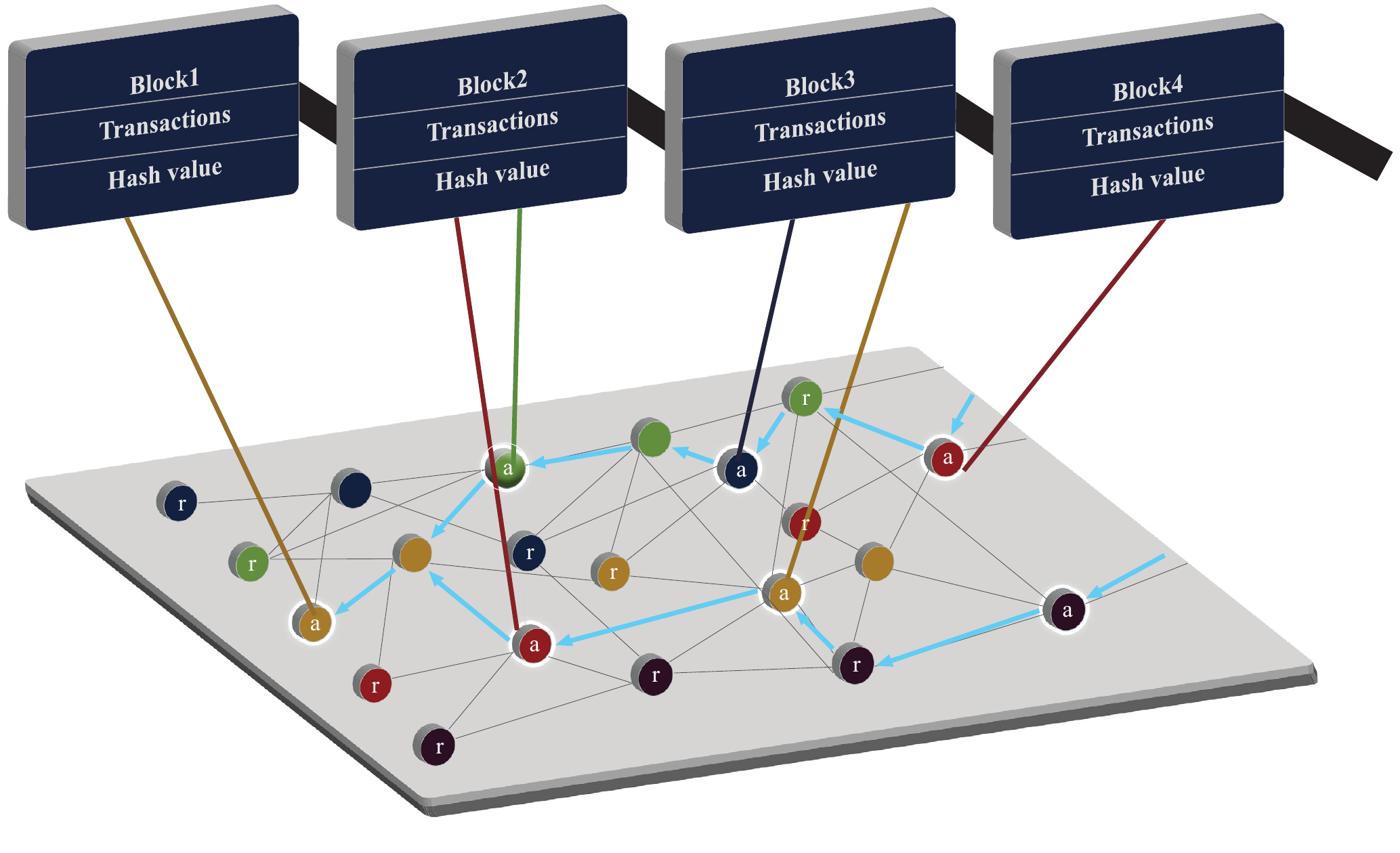}
	\caption{An Example of Main-chain}
	\label{fig:mainchain}
\end{figure}

Each participant node has a copy of the Main chain and can search consensus position of its own event blocks.
Each event block can compute its own consensus position by checking the nearest Atropos event block. Assigning and searching consensus position are introduced in the consensus time selection section.

The Main chain provides quick access to the previous transaction history to efficiently process new coming event blocks. From Main chain,  information about unknown participants or attackers can be easily viewed.
The Main chain can be used efficiently in transaction information management by providing quick access to new event blocks that have been agreed on by the majority of nodes. In short, the Main-chain gives the following advantages:

-	All event blocks or nodes do not need to store all information. It is efficient for data management.

-	Access to previous information is efficient and fast.

Based on these advantages, OPERA chain can respond strongly to efficient transaction treatment and attacks through its Main-chain.

%main procedure
\begin{algorithm}
	\caption{Main Procedure}\label{al:main}
	\begin{algorithmic}[1]
		\Procedure{Main Procedure}{}
		\BState \emph{loop}:
		\State A, B = $k$-node Selection algorithm()
		\State Request sync to node A and B
		\State Sync all known events by Lachesis protocol
		\State Event block creation
		\State (optional) Broadcast out the message
		\State Root selection
		\State Clotho selection
		\State Atropos time consensus
		\BState \emph{loop}:
		\State Request sync from a node
		\State Sync all known events by Lachesis protocol
		\EndProcedure
	\end{algorithmic}
\end{algorithm}

\subsection{Lachesis Consensus Algorithm (LCA)} 

Our Lachesis algorithm (LCA) is presented. LCA is one of the consensus algorithms for solving the byzantine agreement problem. 
In LCA, the OPERA chain uses root, Clotho and Atropos blocks to find consensus time for event blocks. Algorithm~\ref{al:main} shows the pseudo algorithm of a OPERA chain. The algorithm consists of two parts and runs them in parallel.

- In one part, each node requests synchronization and creates an event block. In line 3, a node runs the Node Selection Algorithm. The Node Selection Algorithm returns the $k$ IDs of other nodes to communicate with. In line 4 and 5, the node synchronizes the OPERA chain with other nodes. Line 6 runs the Event block creation, at which step the node creates an event block and checks whether it is root. Then the node broadcasts the created event block to other nodes in line 7. The step in this line is optional. In line 8 and 9, Clotho selection and Atropos time consensus algorithms are invoked. The algorithms determinte whether the specified root can be a Clotho, assign the consensus time, and then confirm the Atropos. 

- The second part is to respond to synchronization requests. In line 10 and 11, the node receives a synchronization request and then sends its response about the OPERA chain.

\subsection{Node Structure}
This section gives an overview of node structure in Lachesis.

Each node has a signature stamp, height vector, in-degree vector, flag table, root hash list, and Main-chain. Signature stamp is the data structure for storing the hash value that indicates the most recently created event block by the node. We call the most recently created event block the top event block. The flag table is a n dimensional vector. If an event block $e$ created by $i^{th}$ node can reach $j^{th}$ root, then the $j^{th}$ value in the flag table of $e$ becomes 1 (otherwise 0). Each node only maintains the flag table of the top event block.

\begin{figure} \centering  
	\includegraphics[width=.4\textwidth]{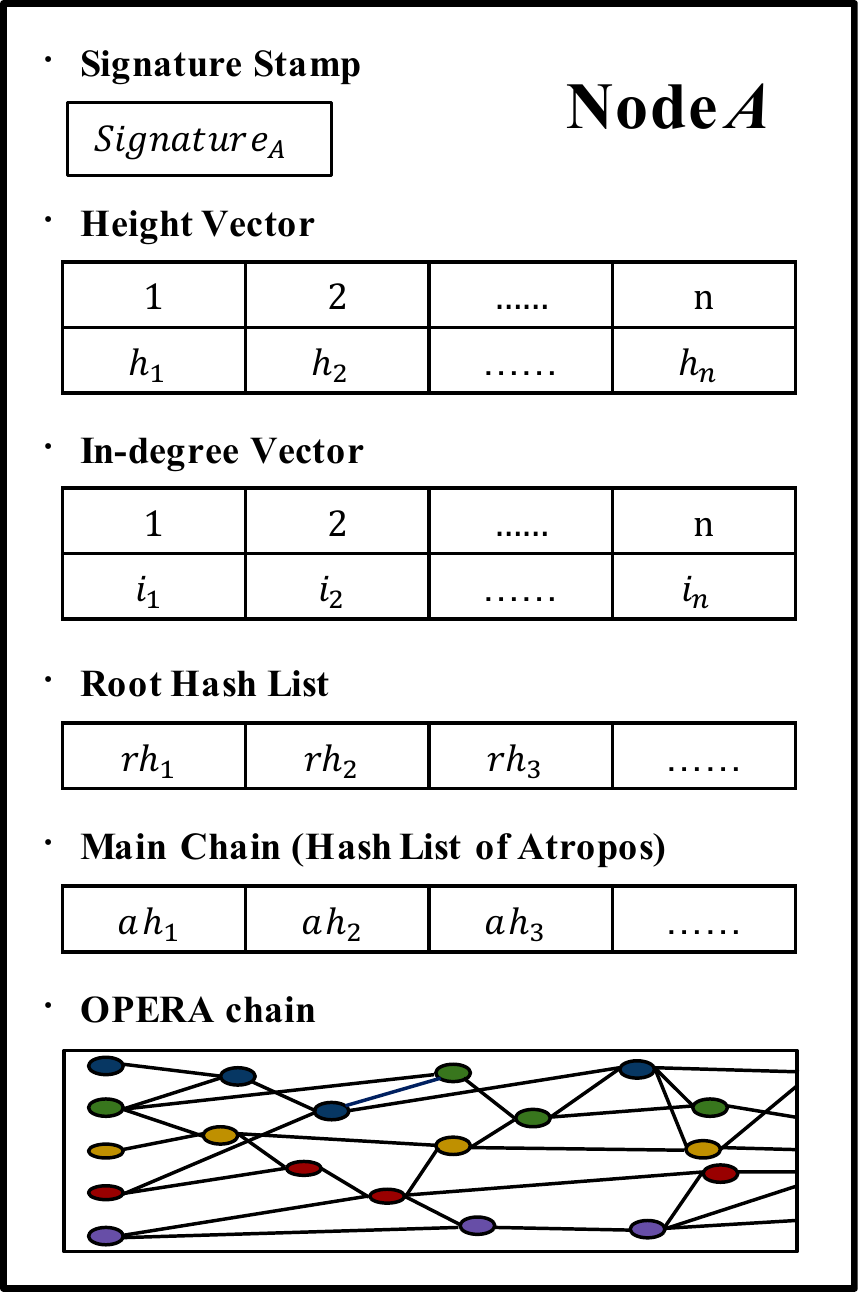}
	\caption{An Example of Node Structure}
	\label{fig:node}
\end{figure}

Figure~\ref{fig:node} shows an example of the node structure component of a node $A$. In the figure, the $signature_A$ stores the hash value of the top event block of $A$. Each value in the height vector is the number of event blocks created by other nodes respectively. The value of $h_i$ is the number of event blocks created by the $i^{th}$ node. Each value in the in-degree vector is the number of edges from other event blocks created by other nodes to the top event block. The root hash list is the data structure storing the hash values of the root. The Main-chain is a data structure storing hash values of the Atropos blocks. The Main-chain is used to find event blocks with complete consensus. The root, Clotho and Atropos selection algorithm are introduced in Section~\ref{se:lca}. 

\subsection{Event block creation}
In Lachesis protocol, every node can create an event block. Each event block refers to other event blocks. Reference means that the event block stores the hash values of the other event blocks. In a Lachesis protocol, an event block refers to $k$-neighbor event blocks under the conditions as follows: 

\begin{enumerate}
\item The reference event blocks are the top event blocks.
\item One reference should be made to a self-parent.
\item The own top event block refers to at least $k$-neighbor of other nodes.
\end{enumerate}

\subsection{Topological ordering of events using Lamport timestamps}
Every node has a physical clock and it needs physical time to create an even block. However, for consensus, Lachesis protocols relies on a logical clock for each node. For the purpose, we use \textit{"Lamport timestamps"} \cite{lamport1978time} to determine the time ordering between event blocks in a asynchronous distributed system.
\\
The Lamport timestamps algorithm is as follows:
\begin{enumerate}
\item Each node increments its count value before creating an event block.
\item When sending a message include its count value, receiver should consider which sender’s message is received and increments its count value.
\item If current counter is less than or equal to the received count value from another node, then the count value of the recipient is updated.  
\item If current counter is greater than the received count value from another node, then the current count value is updated.
\end{enumerate}

\begin{figure}\centering  
	\includegraphics[width=0.9\columnwidth]{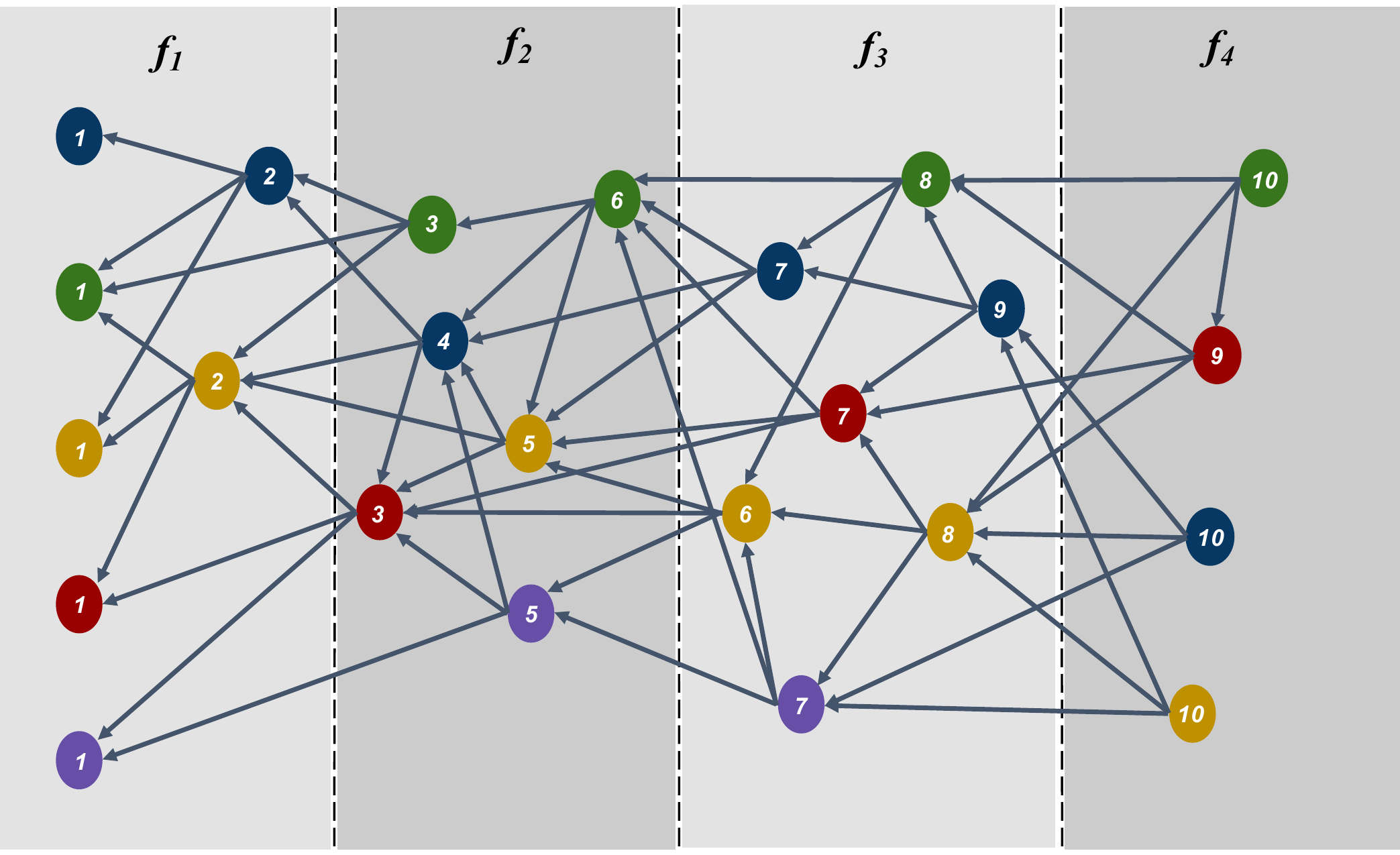}
	\caption{An example of Lamport timestamps}
	\label{fig:Lamport}
\end{figure}

We use the Lamport's algorithm to enforce a topological ordering of event blocks and uses it in Atropos selection algorithm. 

Since an event block is created based on logical time, the sequence between each event blocks is immediately determined. Because the Lamport timestamps algorithm gives a partial order of all events, the whole time ordering process can be used for Byzantine fault tolerance.

\subsection{Topological consensus ordering}
The sequential order of each event block is an important aspect for Byzantine fault tolerance. In order to determine the pre-and-post sequence between all event blocks, we use Atropos consensus time, Lamport timestamp algorithm and the hash value of the event block.

First, when each node creates event blocks, they have a logical timestamp based on Lamport timestamp. This means that they have a partial ordering between the relevant event blocks. 
Each Clotho has consensus time to the Atropos. This consensus time is computed based on the logical time nominated from other nodes at the time of the 2n/3 agreement.

In the LCA, each event block is based on the following three rules to reach an agreement:

\begin{enumerate}
\item If there are more than one Atropos with different times on the same frame, the event block with smaller consensus time has higher priority.
\item If there are more than one Atropos having any of the same consensus time on the same frame, determine the order based on the own logical time from Lamport timestamp.
\item When there are more than one Atropos having the same consensus time, if the local logical time is same, a smaller hash value is given priority through hash function.
\end{enumerate}

\begin{figure}[h] \centering  
	\includegraphics[height=8cm, width=1.0\columnwidth]{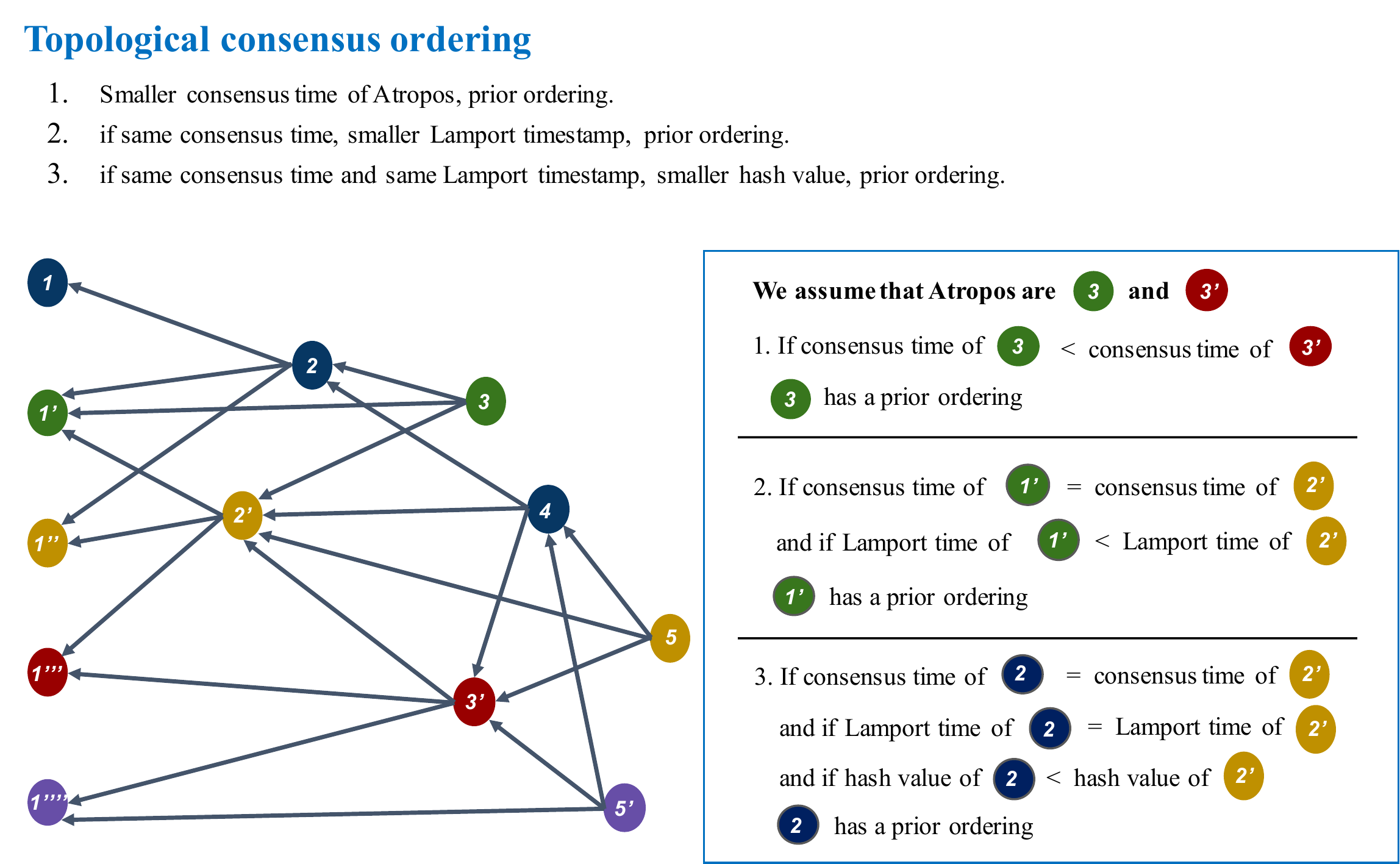}
	\caption{An example of topological consensus ordering}
	\label{fig:topological consensus ordering}
\end{figure}

Figure~\ref{fig:topological consensus ordering} depicts an example of topological consensus ordering.

\begin{figure} \centering  
	\includegraphics[height=8cm, width=1.0\columnwidth]{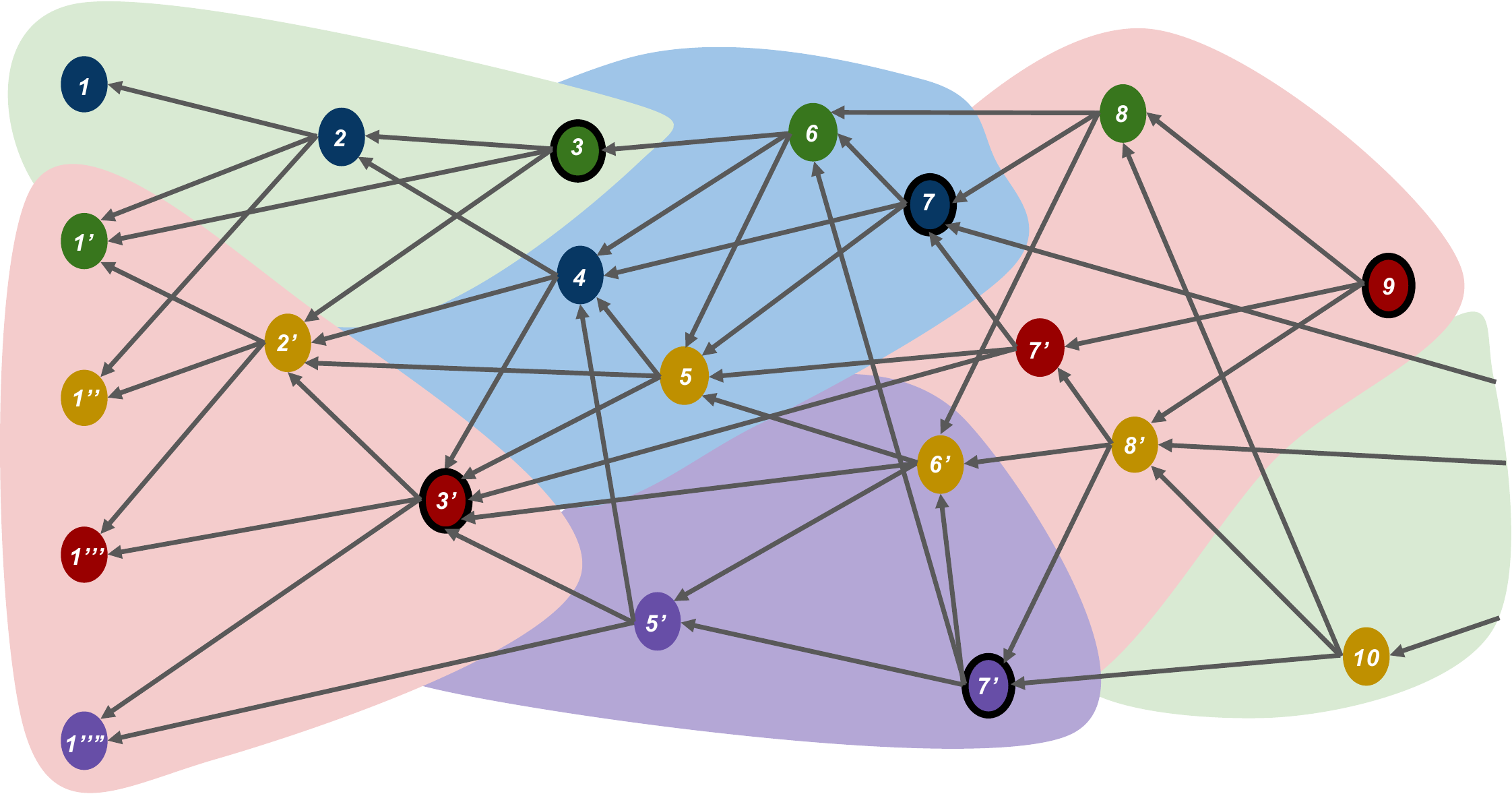}
	\caption{An Example of time ordering of event blocks in OPERA chain}
	\label{fig:sequence of operachain}
\end{figure}

Figure~\ref{fig:sequence of operachain} shows the part of OPERA chain in which the final consensus order is determined based on these 3 rules. The number represented by each event block is a logical time based on Lamport timestamp.

Final topological consensus order containing the event blocks based on agreement for the apropos. Based on each Atropos, they will have different colors depending on their range.

\subsection{Peer selection algorithm}
In order to create an event block, a node needs to select $k$ other nodes. Lachesis protocols does not depend on how peer nodes are selected. One simple approach can use a random selection from the pool of $n$ nodes. The other approach is to define some criteria or cost function to select other peers of a node. 

Within distributed system, a node can select other nodes with low communication costs, low network latency, high bandwidth and high successful transaction throughputs.

\section{Lachesis Consensus Protocol $L_0$}\label{se:lca}
This section presents our new Lachesis Consensus Protocol $L_0$, which is a specific example of the Lachesis class.
We describe the main ideas and algorithms used in the protocol.

\subsection{Root Selection}
All nodes can create event blocks and an event block can be a root when satisfying specific conditions. Not all event blocks can be roots. First, the first created event blocks are themselves roots. These leaf event blocks form the first root set $R_{S1}$. If there are total $n$ nodes and these nodes create the event blocks, then the cardinality of the first root set  $|R_{S1}|$ is $n$. Second, if an event block $e$ can reach at least 2n/3 roots, then $e$ is called a root. This event $e$ does not belong to $R_{S1}$, but the next root set $R_{S2}$. Thus, excluding the first root set, the range of cardinality of root set $R_{Sk}$ is $2n/3 < |R_{Sk}| \leq n$. The event blocks including $R_{Sk}$ before $R_{Sk+1}$ is in the frame $f_k$. The roots in $R_{Sk+1}$ does not belong to the frame $f_k$. Those are included in the frame $f_k+1$ when a root belonging to $R_{Sk+2}$ occurs.  

We introduce the use of a flag table to quickly determine whether a new event block becomes a root or not. Each node maintains a flag table of the top event block.
Every event block that is newly created is assigned $k$ hashes for its $k$ parent event blocks. We apply an $OR$ operation on the flag tables of the parent event blocks.

Figure~\ref{fig:ex_ft} shows an example of how to use flag tables to determine a root. In this example, $r_1$ is the most recently created event block. We apply an $OR$ operation on the flag tables of $r_1$'s $k$ parent event blocks. The result is the flag table of $r_1$.  If $r_1$'s flag table has more than $2n/3$ set bits, $r_1$ is a root. In this example, the number of set bits is 4, which is greater than $2n/3$ ($n$=5). Thus, $r_1$ becomes root.

The root selection algorithm is as follows:

\begin{figure} [t] \centering  
\includegraphics[height=8cm, width=1.0\columnwidth]{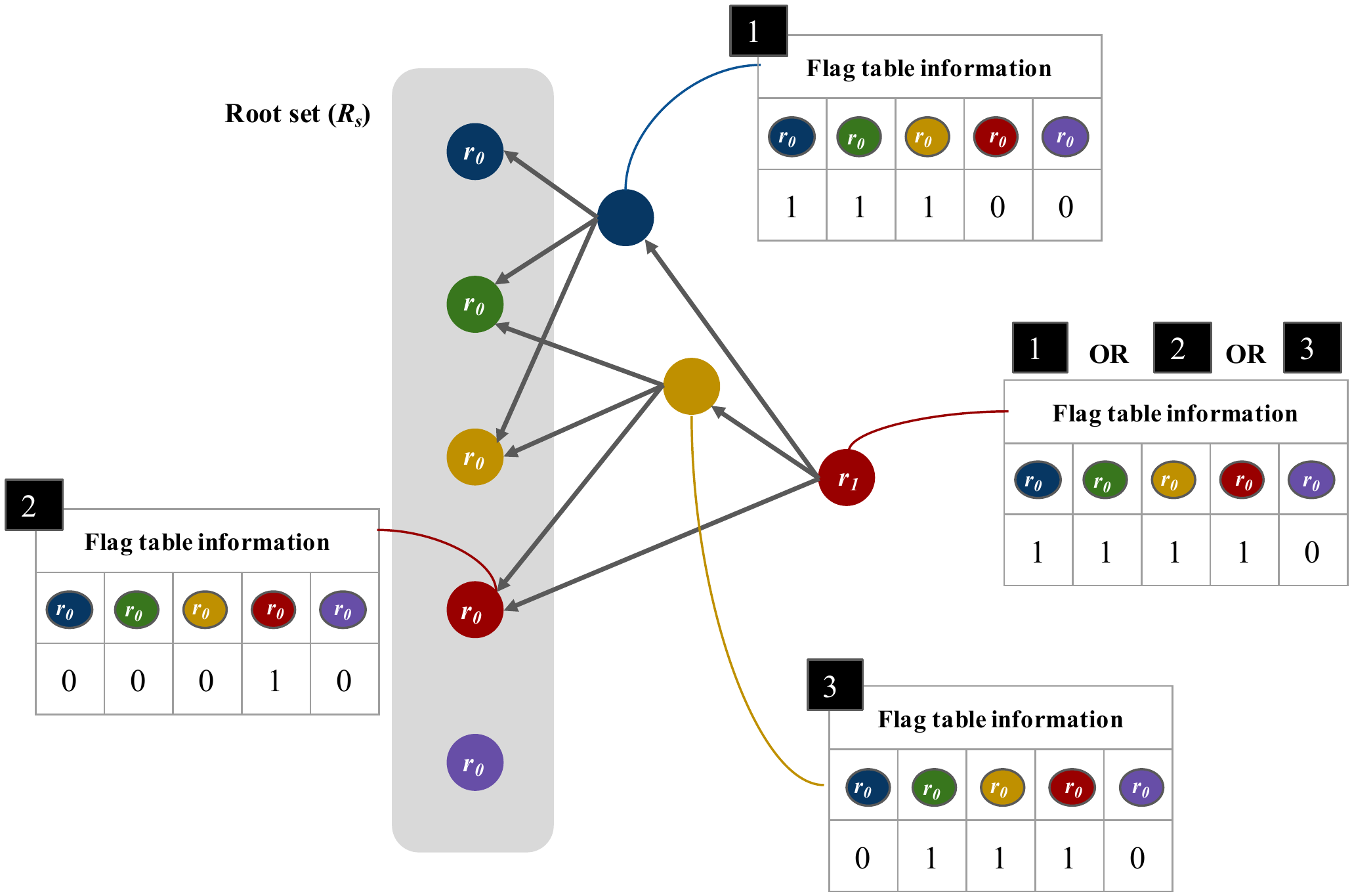}
\caption{An Example of Flag Table Calculation}
\label{fig:ex_ft}
\end{figure}

\begin{enumerate}
\item The first event blocks are considered as root. 
\item When a new event block is added in the OPERA chain, we check whether the event block is a root by applying $OR$ operation on the flag tables connected to the new event block. If the sum of the flag table for the new event block is more than 2n/3, the new event block becomes a root. 
\item When a new root appears on the OPERA chain, nodes update their root hash list. If one of new event blocks becomes a root, all nodes that share the new event block add the hash value of the event block to their root hash list.  
\item The new root set is created if the cardinality of previous root set $R_{Sp}$ is more than 2n/3 and the new event block can reach 2n/3 root in $R_{S_p}$.
\item When the new root set $R_{S_{k+1}}$ is created, the event blocks from previous root set $R_{S_k}$ to before $R_{S_{k+1}}$ belong to the frame $f_k$.
\end{enumerate}

\subsection{Clotho Selection}

A Clotho is a root that satisfies the Clotho creation conditions. 
Clotho creation conditions are that more than 2n/3 nodes know the root and a root knows this information. 

Figure~\ref{fig:Clotho} shows an example of Clotho. Circles with a label $r_i$ (or $c$) represents a root (or Clotho) event block. If there are three other sets of root and there exists one root after the recent clotho set, then one of the roots in the first root set become Clotho.

\begin{figure} [t] \centering  
\includegraphics[height=8cm, width=1.0\columnwidth]{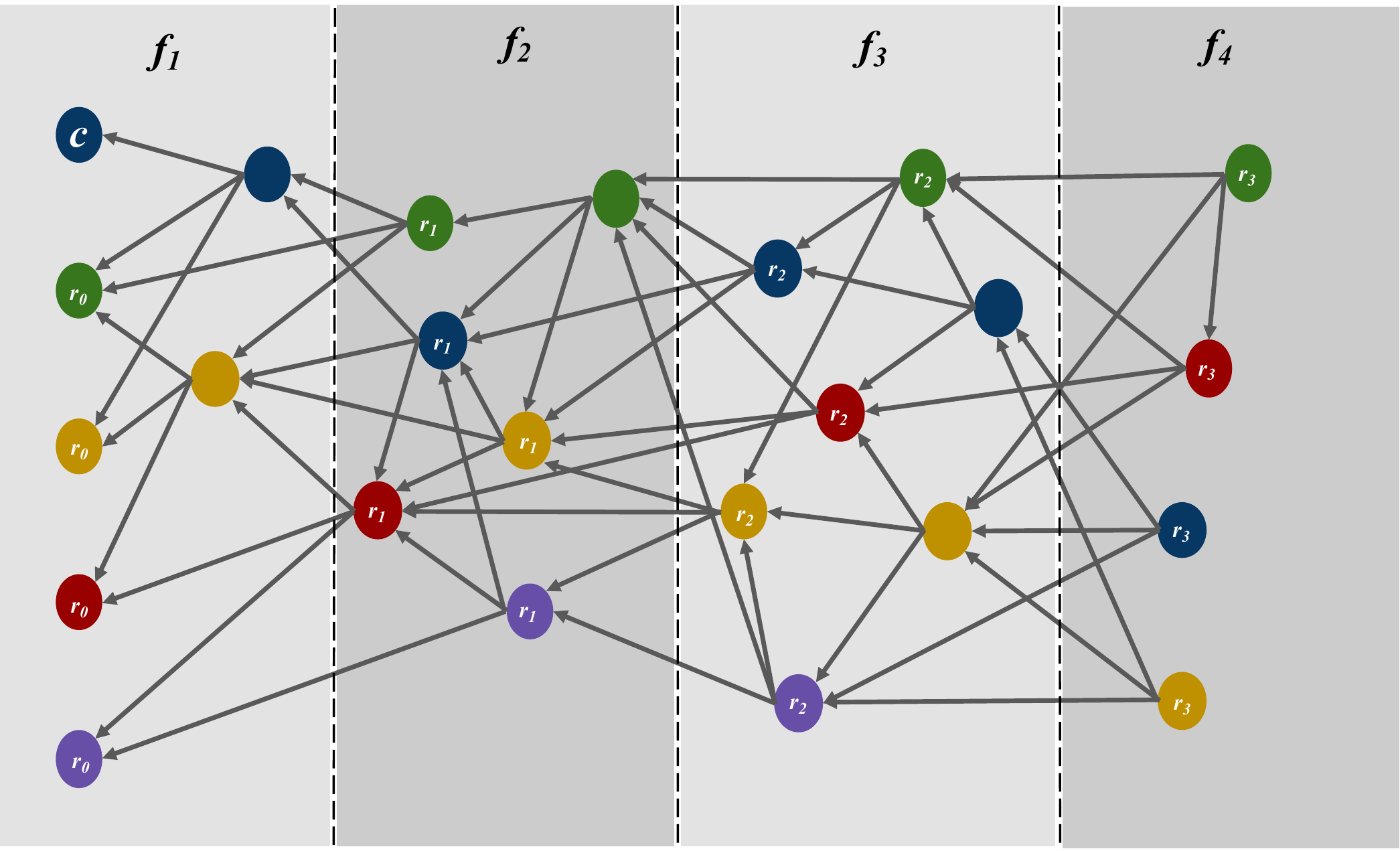}
\caption{An Example of Clotho}
\label{fig:Clotho}
\end{figure}

Clotho selection algorithm checks whether root event blocks in the root hash list satisfy the Clotho condition.
If a root satisfies Clotho condition, the root becomes Clotho and makes a candidate time for Atropos. After the root is concluded as a Clotho, Atropos consensus time selection algorithm is triggered. 

For a root $r$, we denote $frame(i, r)$ to be the root $r$ in $i$-th frame. For example, $frame(1, r)$ is the first root belong to the frame $f_1$. 

Algorithm~\ref{al:acs} shows the pseudo code for Clotho selection. The algorithm takes a root $r$ as input. 
Line 4 and 5 set $c.is\_clotho$ and $c.yes$ to $nil$ and 0 respectively. Line 6-8 checks whether any root $c'$ in $frame(i-3,r)$ shares $c$ where $i$ is the current frame. In line 9-10, if number of roots in $frame(i-2,r)$ which shares $c$ is more than $2n/3$, the root $c$ is set as a Clotho. The time complexity of Algorithm 3 is $O(n^{2})$, where $n$ is the number of nodes. 

%Clotho selection
\begin{algorithm}
\caption{Clotho Selection}\label{al:acs}
\begin{algorithmic}[1]
	\Procedure{Clotho Selection}{}
	\State \textbf{Input}: a root $r$
	\For{$c$ $\in$ $frame(i-3, r)$} 
	\State$c.is\_clotho$ $\leftarrow$ $nil$ 
	\State$c.yes$ $\leftarrow$ 0
	\For{$c'$ $\in$ $frame(i-2, r)$} 
	\If{$c'$ share $c$} 
	\State c.yes $\leftarrow$ c.yes + 1
	\EndIf
	\EndFor
	\If{$c.yes > 2n/3$}
	\State $c.is\_clotho$ $\leftarrow$ $yes$
	\EndIf
	\EndFor
	\EndProcedure
\end{algorithmic}
\end{algorithm}

\subsection{Atropos Selection}
Atropos selection algorithm is the process in which the candidate time generated from Clotho selection is shared with other nodes, and each root re-selects candidate time repeatedly until all nodes have same candidate time for a Clotho. 

After a Clotho is nominated, each node then computes candidate time of the Clotho. If there are more than two-thirds of the nodes that compute the same value for candidate time, that time value is recorded. Otherwise, each node reselects candidate time from some candidate time which the node collects. By the reselection process, each node reaches time consensus for candidate time of Clotho as OPERA chain grows. The candidate time reaching the consensus is called Atropos consensus time. After Atropos consensus time is computed,  Clotho is nominated to Atropos and each node stores the hash value of Atropos and Atropos consensus time in Main-Chain. The Main-chain is used for time order between event blocks. The proof of Atropos consensus time selection is shown in the section~\ref{se:proof}. 

%atropos time consensus
\begin{algorithm}[H]
\caption{Atropos Consensus Time Selection}\label{al:atc}
\begin{algorithmic}[1]
	\Procedure{Atropos Consensus Time Selection}{}
	\State \textbf{Input}: $c.Clotho$ in frame $f_i$
	\State$c.consensus\_time$ $\leftarrow$ $nil$
	\State$m$ $\leftarrow$ the index of the last frame $f_m$
	\For{d from 3 to (m-i)}
	\State $R$ $\leftarrow$ be the Root set $R_{S_{i+d}}$ in frame $f_{i+d}$
	\For{$r$ $\in$ $R$} 
	\If{d is 3}
	\If{$r$ confirms $c$ as Clotho}
	\State $r.time(c)$ $\leftarrow$ $r.lamport\_time$
	\EndIf
	\ElsIf{d $>$ 3}
	\State s $\leftarrow$ the set of Root in $f_{j-1}$ that $r$ can share
	\State t $\leftarrow$ RESELECTION(s, $c$)
	\State k $\leftarrow$ the number of root having $t$ in $s$
	\If{d mod $h$ $>$ 0}
	\If{$k$ $>$ 2n/3}
	\State $c.consensus\_time$ $\leftarrow$ $t$
	\State $r.time(c)$ $\leftarrow$ $t$
	\Else
	\State $r.time(c)$ $\leftarrow$ $t$
	\EndIf
	\Else
	\State $r.time(c)$ $\leftarrow$ the minimum value in $s$
	\EndIf
	\EndIf
	\EndFor
	\EndFor
	\EndProcedure
\end{algorithmic}
\end{algorithm}

Algorithm~\ref{al:atc} and~\ref{al:resel} show pseudo code of Atropos consensus time selection and Consensus time reselection. In Algorithm~\ref{al:atc}, at line 6, $d$ saves the deference of relationship between root set of $c$ and $w$. Thus, line 8 means that $w$ is one of the elements in root set of the frame $f_{i+3}$, where the frame $f_i$ includes $c$. Line 10, each root in the frame $f_j$ selects own Lamport timestamp as candidate time of $c$ when they confirm root $c$ as Cltoho. In line 12, 13, and 14, $s$, $t$, and $k$ save the set of root that $w$ can share $c$, the result of $RESELECTION$ function, and the number of root in $s$ having $t$. Line 15 is checking whether there is a difference as much as $h$ between $i$ and $j$ where $h$ is a constant value for minimum selection frame. Line 16-20 is checking whether more than two-thirds of root in the frame $f_{j-1}$ nominate the same candidate time. If two-thirds of root in the frame $f_{j-1}$ nominate the same candidate time, the root $c$ is assigned consensus time as $t$. Line 22 is minimum selection frame. In minimum selection frame, minimum value of candidate time is selected to reach byzantine agreement. Algorithm~\ref{al:resel} operates in the middle of Algorithm~\ref{al:atc}. In Algorithm~\ref{al:resel}, input is a root set $W$ and output is a reselected candidate time. Line 4-5 computes the frequencies of each candidate time from all the roots in $W$. In line 6-11, a candidate time which is smallest time that is the most nomitated. The time complexity of Algorithm~\ref{al:resel} is $O(n)$ where $n$ is the number of nodes. Since Algorithm~\ref{al:atc} includes Algorithm~\ref{al:resel}, the time complexity of Algorithm~\ref{al:atc} is $O(n^2)$ where $n$ is the number of nodes.

%reselection
\begin{algorithm} [H]
\caption{Consensus Time Reselection}\label{al:resel}
\begin{algorithmic}[1]
	\Function{Reselection}{} \funclabel{alg:a}
	\State \textbf{Input}: Root set $R$, and Clotho $c$
	\State \textbf{Output}: candidate time $t$
	\State $\tau$ $\leftarrow$ set of all $t_i = r.time(c)$ for all $r$ in $R$
	\State $D$  $\leftarrow$ set of tuples $(t_i, c_i)$ computed from $\tau$, where $c_i= count(t_i)$
	\State $max\_count$ $\leftarrow$ $max(c_i)$
	\State $t$ $\leftarrow$ $infinite$
	\For{tuple $(t_i, c_i)$ $\in$ $D$}
	\If{$max\_count$ $==$ $c_i$ $\&\&$ $t_i$ $<$ $t$}
	\State $t$ $\leftarrow$ $t_i$
	\EndIf
	\EndFor
	\State \textbf{return} $t$
	\EndFunction
\end{algorithmic}
\end{algorithm}

In the Atropos Consensus Time Selection algorithm, nodes reach consensus agreement about candidate time of a Clotho without additional communication (i.e., exchanging candidate time) with each other. Each node communicates with each other through the Lachesis protocol, the OPERA chain of all nodes grows up into same shape. This allows each node to know the candidate time of other nodes based on its OPERA chain and reach a consensus agreement. The proof that the agreement based on OPERA chain become agreement in action is shown in the section~\ref{se:proof}.\\

\subsection{Peer selection algorithm via Cost function}
We define three versions of the Cost Function ($C_F$). Version one is focused around updated information share and is discussed below. The other two versions are focused on root creation and consensus facilitation, these will be discussed in a following paper.

We define a Cost Function ($C_F$) for preventing the creation of lazy nodes. The lazy node is a node that has a lower work portion in the OPERA chain. When a node creates an event block, the node selects other nodes with low values outputs from the cost function and refers to the top event blocks of the reference nodes. An equation~(\ref{eq1}) of $C_F$ is as follows,

\begin{equation}\label{eq1}
C_{F} =I/H
\end{equation}

where $I$ and $H$ denote values of in-degree vector and height vector respectively. If the number of nodes with the lowest $C_F$ is more than $k$, one of the nodes is selected at random. The reason for selecting high $H$ is that we can expect a high possibility to create a root because the high $H$ indicates that the communication frequency of the node had more opportunities than others with low $H$. Otherwise, the nodes that have high $C_F$ (the case of $I$ $>$ $H$) have generated fewer event blocks than the nodes that have low $C_F$. Then we can judge that those kind of nodes are lazy. If we can detect whether a node is lazy based on  cost function, we can change the lazy nodes to other participants or remove them. 

\begin{figure}[htp] \centering  
\includegraphics[height=6cm]{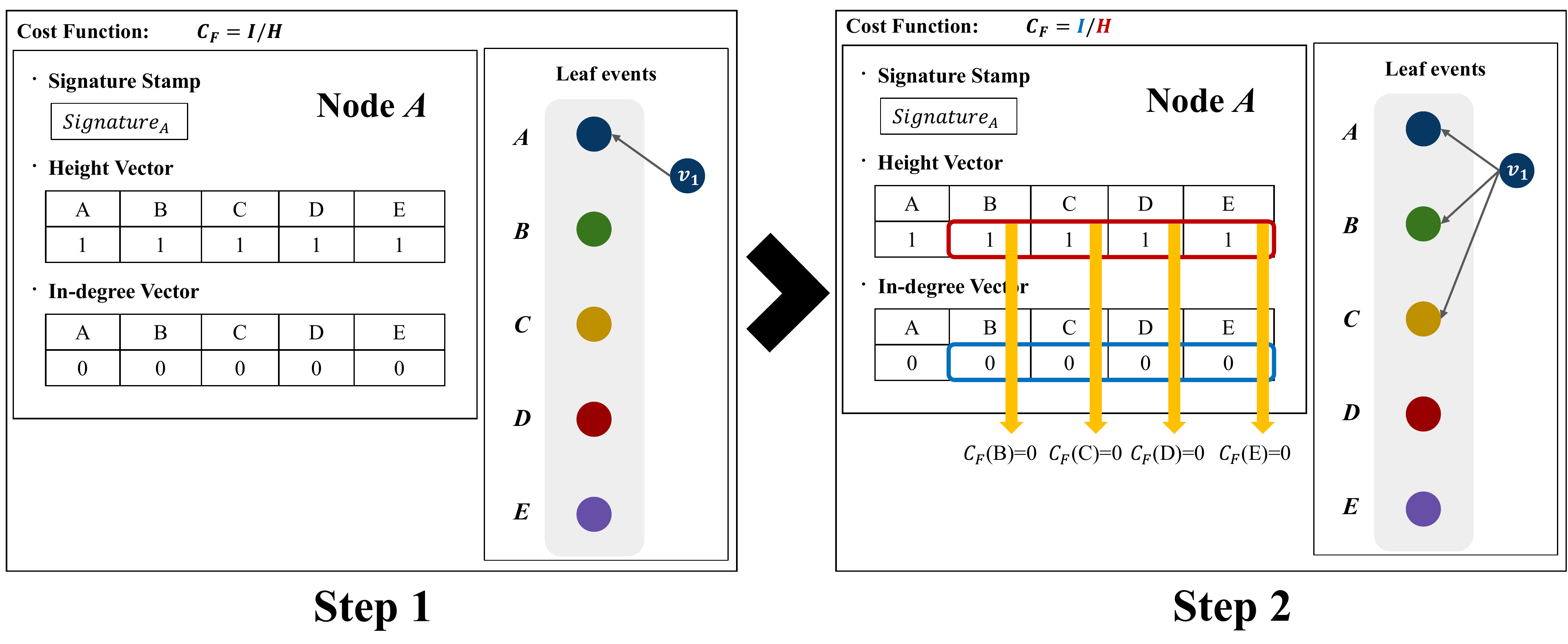}
\caption{An Example of Cost Function 1}
\label{fig:cfEx1}
\end{figure}

Figure~\ref{fig:cfEx1} shows an example of the node selection based on the cost function after the creation of leaf events by all nodes. In this example, there are five nodes and each node created leaf events. All nodes know other leaf events. Node $A$ creates an event block $v_1$ and $A$ calculates the cost functions. Step 2 in Figure~\ref{fig:cfEx1} shows the results of cost functions based on the height and in-degree vectors of node $A$. In the initial step, each value in the vectors are same because all nodes have only leaf events. Node $A$ randomly selects $k$ nodes and connects $v_1$ to the leaf events of selected nodes. In this example, we set $k$=3 and assume that node $A$ selects node $B$ and $C$. 

\begin{figure} \centering  
\includegraphics[height=7cm]{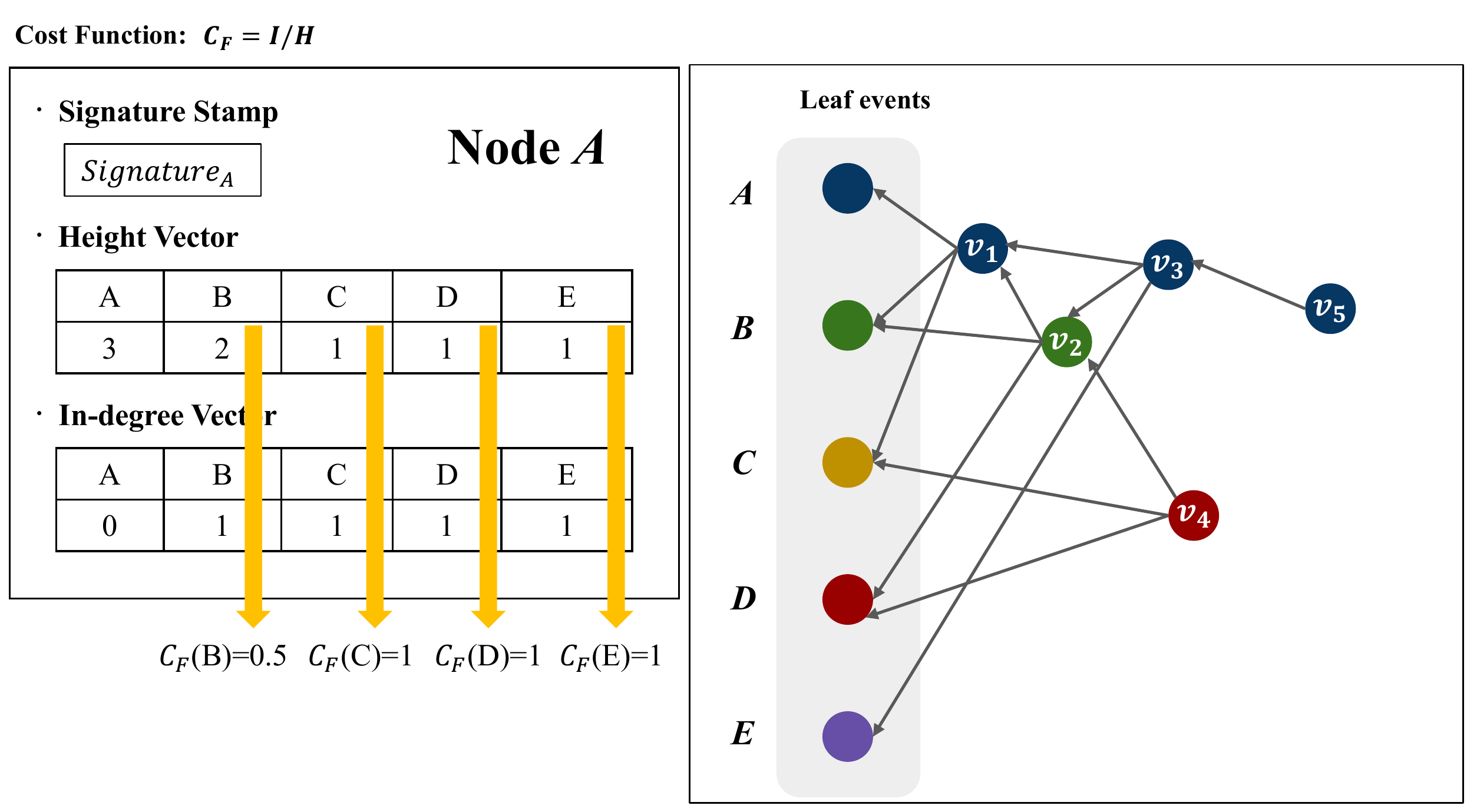}
\caption{An Example of Cost Function 2}
\label{fig:cfEx2}
\end{figure}

Figure~\ref{fig:cfEx2} shows an example of the node selection after a few steps of the simulation in Figure~\ref{fig:cfEx1}. In Figure~\ref{fig:cfEx2}, the recent event block is $v_5$ created by node $A$. Node $A$ calculates the cost function and selects the other two nodes that have the lowest results of the cost function. In this example, node $B$ has 0.5 as the result and other nodes have the same values. Because of this, node $A$ first selects node $B$ and randomly selects other nodes among nodes $C$, $D$, and $E$.

The height of node $D$ in the current OPERA chain of the example is 2 (leaf event and event block $v_4$). On the other hand, the height of node $D$ in node structure of $A$ is 1. Node $A$ is still not aware of the presence of the event block $v_4$. It means that there is no path from the event blocks created by node $A$ to the event block $v_4$. Thus, node $A$ has 1 as the height of node $D$. 

Algorithm~\ref{al:ns} shows the selecting algorithm for selecting reference nodes. The algorithm operates for each node to select a communication partner from other nodes. Line 4 and 5 set min\_cost and $S_{ref}$ to initial state. Line 7 calculates the cost function $c_f$ for each node. In line 8, 9, and 10, we find the minimum value of the cost function and set min\_cost and $S_{ref}$ to $c_f$ and the ID of each node respectively. Line 11 and 12 append the ID of each node to $S_{ref}$ if min\_cost equals $c_f$. Finally, line 13 selects randomly $k$ node IDs from $S_{ref}$ as communication partners. The time complexity of Algorithm 2 is $O(n)$, where $n$ is the number of nodes. 

%node selection
\begin{algorithm}
\caption{$k$-neighbor Node Selection}\label{al:ns}
\begin{algorithmic}[1]
	\Procedure{$k$-node Selection}{}
	\State \textbf{Input:} Height Vector $H$, In-degree Vector $I$
	\State \textbf{Output:} reference node $ref$
	\State min\_cost $\leftarrow$ $INF$ 
	\State $s_{ref}$ $\leftarrow$ None
	\For{$k \in Node\_Set$}
	\State $c_f$ $\leftarrow$ $\frac{I_k}{H_k}$ 
	\If {min\_cost $>$ $c_f$} 
	\State min\_cost $\leftarrow$ $c_f$
	\State $s_{ref}$ $\leftarrow$ {k}
	\ElsIf {min\_cost $equa$l $c_f$}
	\State $s_{ref}$ $\leftarrow$ $s_{ref}$ $\cup$ $k$
	\EndIf
	\EndFor
	\State $ref$ $\leftarrow$ random select in $s_{ref}$
	\EndProcedure
\end{algorithmic}
\end{algorithm}

After the reference node is selected, each node communicates and shares information that is all event blocks known by them. A node creates an event block by referring to the top event block of the reference node. The Lachesis protocol works and communicates asynchronously. This allows a node to create an event block asynchronously even when another node creates an event block. The communication between nodes does not allow simultaneous communication with the same node. 

\begin{figure} \centering  
\includegraphics[height=5cm]{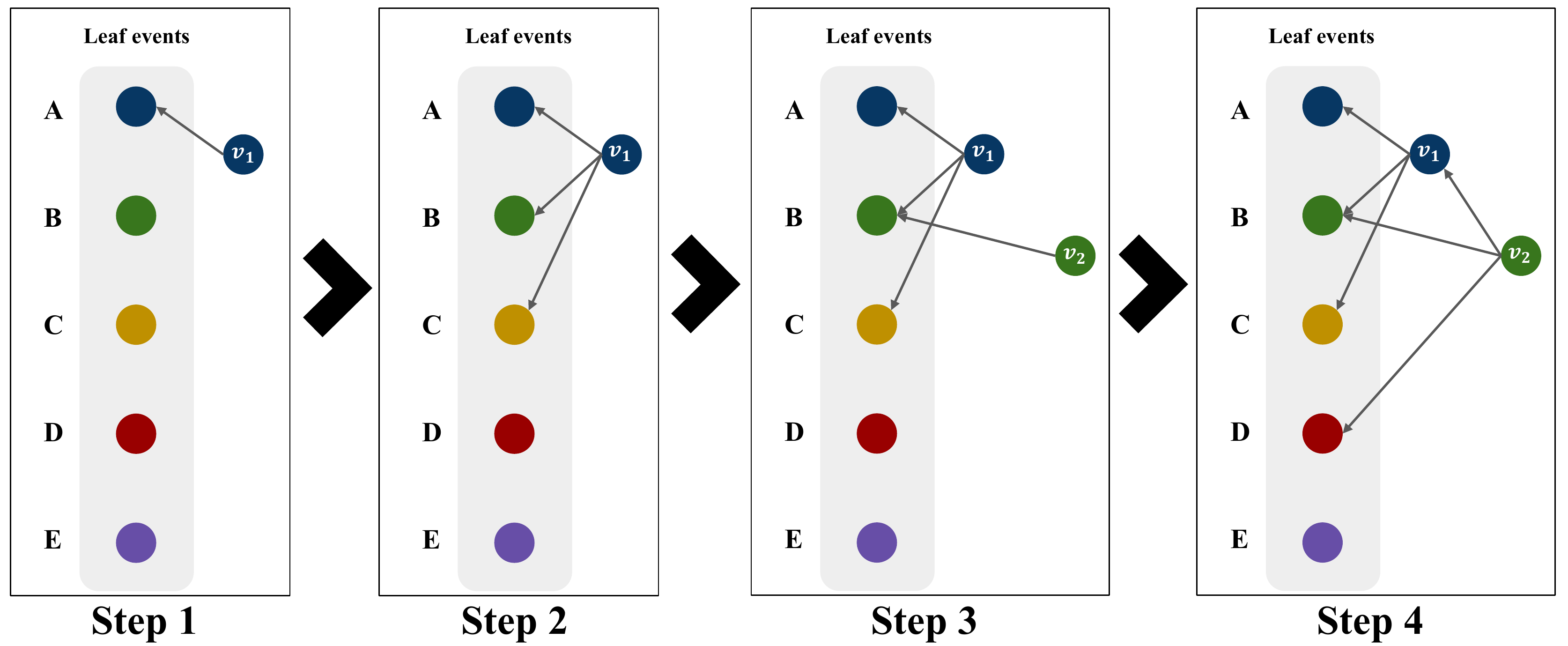}
\caption{An Example of Node Selection}
\label{fig:nsEx}
\end{figure}

Figure~\ref{fig:nsEx} shows an example of the node selection in Lachesis protocol. In this example, there are five nodes ($A, B, C, D,$ and $E$) and each node generates the first event blocks, called leaf events. All nodes share other leaf events with each other. In the first step, node $A$ generates new event block $v_1$ (\textit{blue}). Then node $A$ calculates the cost function to connect other nodes. In this initial situation, all nodes have one event block called leaf event, thus the height vector and the in-degree vector in node $A$ has same values. In other words, the heights of each node are 1 and in-degrees are 0. Because of this reason, node $A$ randomly select other two nodes and connect $v_1$ to the top two event blocks by other two nodes. The step 2 shows the situation after connections. In this example, node $A$ select node $B$ and $C$ to connect $v_1$ and the event block $v_1$ is connected to the top event blocks of node $B$ and $C$. Node $A$ only knows the situation of the step 2. 

After that, in the example, node $B$ generates new event block $v_2$ (\textit{green}) and also calculates the cost function. $B$ randomly select the other two nodes; $A$, and $D$, since $B$ only has information of the leaf events. Node $B$ requests to $A$ and $D$ to connect $e_B$, then nodes $A$ and $D$ send information for their top event blocks to node $B$ as response. The top event block of node $A$ is $v_1$ and node $D$ is the leaf event. The event block $v_2$ is connected to $v_1$ and leaf event from node $D$. Step 4 shows these connections.

\section{Discussions }\label{se:discuss}
This section presents several discussions on our Lachesis protocol.

\subsection{Lamport timestamps}
This section discusses a topological order of event blocks in DAG-based Lachesis protocols using Lamport timestamps~\cite{lamport1978time}.

Our Lachesis protocols relies on Lamport timestamps to define a topological ordering of event blocks in OPERA chain.
The ``happened before" relation, denoted by $\rightarrow$, gives a partial ordering of events from a distributed system of nodes. 

Given $n$ nodes, they are represented by $n$ processes $P = (P_0, P_1, \dots, P_{n-1})$. 

For a pair of event blocks $b$ and $b'$, the relation "$\rightarrow$" satisfies: (1) If $b$ and $b'$ are events of process $P_i$, and $b$ comes before $b'$, then $b \rightarrow b'$. (2) If $b$ is the send($m$) by one process and $b'$ is the receive($m$) by another process, then $b \rightarrow b'$. (3) If $b \rightarrow b'$ and $b' \rightarrow b''$ then $b \rightarrow b''$. 
Two distinct events $b$ and $b'$ are said to be concurrent if $b \nrightarrow b'$ and $b' \nrightarrow b$.

For an arbitrary total ordering $\prec$ of the processes, a relation $\Rightarrow$ is defined as follows: if $b$ is an event in process $P_i$ and $b'$ is an event in process $P_j$, then $b \Rightarrow b'$ if and only if either (i) $C_i(a) < C_j(b)$ or (ii) $C(b)= C_j(b')$ and $P_i \prec P_j$. This defines a total ordering, and that the Clock Condition implies that if $a \rightarrow b$ then $a \Rightarrow b$. 

We use this total ordering in our Lachesis algorithms. 
By using Lamport timestamps, we do not rely on physical locks to determine a partial ordering of events. 

\subsection{Semantics of Lachesis protocols}

This section discusses the possible usage of concurrent common knowledge, described in Section~\ref{se:cck} to understand DAG-based consensus protocols.

Let $G=(V, E)$ denote directed acyclic graph (DAG). $V$ is a set of vertices and $E$ is a set of edges. DAG is a directed graph with no cycle. Namely, in DAG, there is no path that source and destination at the same vertex. 
A path is a sequence of vertices ($v_1$, $v_2$, ..., $v_\textit{(k-1)}$, $v_k$) that uses no edge more than once.

An asynchronous system consists of the following sets.
\begin{enumerate}
\item A set $P$ = \{1,...,$n$\} of process identifiers, where $n$ is the total number of processes $P_i$ in the system.
\item A set $C$ $\subseteq$ \{($i$,$j$) s.t. $i,j \in P$\} of channels. If $(i,j)$ in $C$, it indicates that process $i$ can send messages to process $j$.
\item A set $H_i$ of possible local histories for each process $i$ in $P$.
\item A set $A$ of asynchronous runs. Each asynchronous run is a vector of local histories, denoted by $a = \langle h_1,h_2,h_3,...h_N \rangle$. Each process has a single run. Histories are indexed by process identifiers.
\item A set $M$ of messages. A message is a triple $\langle i,j,B \rangle$ where $i \in P$ is the sender of the message, $j \in $ is the message recipient, and $B$ is the message body.	
\end{enumerate}

In Lachesis protocol, each node selects $k$ other nodes as peers. 
For certain gossip protocol, nodes may be constrained to gossip with its $k$ peers. In such a case, the set of channels $C$ can be modelled as follows.
If node $i$ selects node $j$ as a peer, then $(i,j) \in C$. In general, one can express the history of each node in Lachesis protocol in the same manner as in the CCK paper~\cite{cck92}.
Thus, a proof of consensus can be formalized via the consistent cuts.

\section{Conclusion}\label{se:con}
In order to realize the distributed ledger technology, we have proposed a new family of asynchronous DAG-based consensus protocol, namely $\mathcal{L}$. 
We introduce the OPERA chain and Main-chain for faster consensus. By using Lamport timestamps, the topological ordering of event blocks in OPERA chain and Main chain is more intuitive and reliable in distributed system. We introduce a flag table at each block to improve root detection.

Further, we have presented a specific Lachesis consensus protocol, called $L_0$, as an example of $\mathcal{L}$. The $L_0$ protocol uses a new flag table in each block as a shortcut to check for reachability from an event block to a root along the OPERA chain. The path search is used as a proof of pBFT consensus. In terms of effectiveness, using flag table in $L_0$ protocol is more effective for consensus compared to the path searching approaches. To ensure the distribution of participating nodes, the Lachesis protocol defines a new cost function and an algorithm that efficiently and quickly selects peers. We also propose new algorithms for root selection and Clotho block selection based on the flag table; for Atropos selection by Weight after time consensus ordering.  

Based on the  $L_0$ protocol and the new consensus algorithm, the OPERA chain can protect against malicious attacks such as forks, double spending, parasite chains, and network control. These protections guarantee the safety of OPERA chain. We can also verify existence of Atropos with the OPERA chain. It concludes that the OPERA chain reaches consensus and guarantees liveliness. 
Finally, the time ordering ensures guarantee by weight value on the flag table. Based on these properties, the LCA provides a fair, transparent, and effective consensus algorithm.

\subsection{Future work}
There are a number of directions for future work:

\begin{itemize}
\item With the Lachesis protocols, we are investigating a fast node synchronization algorithm with $k$-neighbor broadcasts.
With OPERA chain and $k$ peer selection, it is possible to achieve a faster gossip broadcast. We are interested in comparing performance of different gossip strategies, such as randomized gossip, broadcast gossip and collection tree protocol for distributed averaging in wireless sensor networks.
\item We are also investigating the semantics of DAG-based protocols in general and Lachesis protocols in particular. We aim to have a formal proof of pBFT using concurrent common knowledge via consistent cuts. 
\end{itemize}

\newpage
\section{Appendix}\label{se:appendix}
\subsection{Proof of Lachesis Consensus Algorithm}\label{se:proof}
In this section, we provide proof of liveness and safety in OPERA chain and show the Byzantine fault tolerance. To represent the Byzantine fault tolerance, we assume that more than two-thirds of participants are reliable nodes. Based on the assumption, we provide some definitions, lemmas and theorems. Then, we eventually validate the Byzantine fault tolerance. 

\subsubsection{Preliminaries}
Let $G=(V, E)$ denote directed acyclic graph (DAG). $V$ is a set of vertices and $E$ is a set of edges. DAG is a directed graph with no cycle. Namely, in DAG, there is no path that source and destination at the same vertex. 
A path is a sequence $P$ of vertices ($v_1$, $v_2$, ..., $v_\textit{(k-1)}$, $v_k$) that uses no edge more than once. Suppose that we have a current vertex $v_c$ and current event block $e_c$ respectively.
A vertex $v_p$ is parent of $v_c$ if there is a path from $v_c$ to $v_p$ and the length of path is 1. 
A vertex $v_a$ is ancestor of $v_c$ if there is a path from $v_c$ to $v_a$ and the length of path is more than equal to 1.

\subsubsection{Proof of Byzantine Fault Tolerance for Lachesis Consensus Algorithm}
\begin{defn}[node]
The machine that participates in the OPERA chain and creates event blocks. The total number of nodes is $n$.
\end{defn}

\begin{defn}[event block]
In OPERA chain, we call a vertex an event block.
\end{defn}

\begin{defn}[self parent]
An event block $v_s$ is self parent of an event block $v_c$ if $v_s$ is parent of $v_c$ and both event blocks have same signatures. 
\end{defn}

\begin{defn}[self ancestor]
An event block $v_a$ is self ancestor of an event block $v_c$ if $v_a$ is ancestor of $v_c$ and both event blocks have same signatures. 
\end{defn}

\begin{defn}[Happened-Before]
An event block $v_x$ Happened-Before an event block $v_y$ if there is a path from $v_x$ to $v_y$.
\end{defn}

\begin{defn}[Root]
The first created event blocks (leaf events) become root or an event block $v$ that can reach more than 2n/3 other roots, becomes a root. 
\end{defn}

\begin{defn}[Root set]
All first event blocks (leaf events) are elements of root set $R_1$ ($|R_1|$ = $n$). And the root set $R_k$ is a set of roots such that $r_i$ $\in$ $R_k$ cannot reach more than 2n/3 other roots in $R_k$ $(k > 1)$.  
\end{defn}

\begin{defn}[Frame]
Frame $f$ is a natural number that separates Root sets. 
\end{defn} 

\begin{defn}[Clotho]
A root $r_k$ in the frame $f_{a+3}$ can nominate a root $r_a$ as Clotho if more than 2n/3 roots in the frame $f_{a+1}$ Happened-Before $r_a$ and $r_k$ Happened-Before the roots in the frame $f_{a+1}$.
\end{defn} 

\begin{defn}[Atropos]
If the consensus time of Clotho is validated, the Clotho become Atropos. 
\end{defn}

\begin{prop}
\label{prop:seen}
At least 2n/3 roots in the frame $f_i$ Happened-Before at least 2n/3 roots in the frame $f_{i+1}$. 
\end{prop}

\begin{proof}
The number of roots in each root set is more than 2n/3. Since a root in the frame $f_{i+1}$ Happened-Before more than 2n/3 roots in the frame $f_i$, when the cardinalities of the root sets in the frames $f_i$ and $f_{i+1}$ are n and 2n/3 respectively, the number of paths from root set in the frame $f_{i+1}$ to root set in the frame $f_{i}$ is at least (2n/3)$^{2}$. The average and the maximum of the number of paths from root set in the frame $f_{i+1}$ to an root in the frame $f_{i}$ are (4n/9) and (2n/3) respectively. Thus, at least 2n/3 roots in the frame $f_{i}$ Happened-Before at least n/3 root in the frame $f_{i+1}$. 
\end{proof}

\begin{prop}
\label{prop:share}
If a root in the frame $f_{i}$ Happened-Before from more than n/3 roots in the frame $f_{i+1}$, the root Happened-Before all roots in the frame $f_{i+2}$.
\end{prop}

\begin{proof}
Based on the definition of Root, each root can reach more than 2n/3 other roots in the previous frame. It means that a root in the frame $f_{i+2}$ should have a number of paths more than 2n/3 to roots in the frame $f_{i+1}$. Thus, if a root $r$ in the frame $f_{i}$ Happened-Before more than n/3 roots in the frame $f_{i+1}$, all roots in the frame $f_{i+2}$ have path to the root $r$.
\end{proof}

\begin{lem}[Sharing]
\label{lem:share}
If a root $r_a$ in the frame $f_{a+3}$ is created, the root in the frame $f_{a+3}$ knows that more than 2n/3 roots in the frame $f_{a}$ become known by more than 2n/3 nodes.
\end{lem}

\begin{proof}
Based on propositions~\ref{prop:seen} and ~\ref{prop:share}, the root in the frame $f_{a+3}$ knows that more than 2n/3 roots in the frame $f_{a}$ become known by more than 2n/3 nodes. 
\end{proof}

\begin{lem}[Fork]
\label{lem:fork}
If the pair of event blocks ($x, y$) is a fork, roots happened-before at least one fork in OPERA chain. Therefore, they can know fork before becoming Clotho.
\end{lem}

\begin{proof}
Suppose that a node creates two event blocks ($x, y$) and the event blocks are a fork. To create two Clotho that can reach each event block in the pair, the event blocks should be shared in more than 2n/3 nodes. Therefore, if there exist fork event blocks, the OPERA chain can structurally detect the fork before roots become Clotho.
\end{proof}

\begin{thm}
\label{thm:same}
All node grows up into same shape in OPERA chain.
\end{thm}

\begin{proof}
Suppose that each node $A$ and $B$ will have a different shape (or a structure). For any two nodes $A$ and $B$, there is two event blocks $x$ and $y$ which are in both $OPERA(A)$ and $OPERA(B)$. Also, their path between $x$ and $y$ in $OPERA(A)$ is not equal to that in $OPERA(B)$. For any two event blocks, if each node has different paths, we can consider that the difference is fork attacks. Based on Lemma~\ref{lem:fork}, if an attacker forks an event block, the OPERA chain can detect and remove it before the Clotho is generated. It contradicts our assumptions. For this reason, two nodes have consistent OPERA chain. 
\end{proof}

\begin{lem}
\label{lem:root}
For any root set $R$, all nodes nominate same root into Clotho.
\end{lem}

\begin{proof}
Based on Theorem~\ref{thm:same}, each node nominates a root into Clotho via the flag table. If all nodes have an OPERA chain with same shape, the values in flag table should be equal to each other in OPERA chain. Thus, all nodes nominate the same root into Clotho since the OPERA chain of all nodes has same shape.
\end{proof}

\begin{lem}
\label{lem:resel}
In the Reselection algorithm, for any Clotho, a root in OPERA chain selects the same consensus time candidate.
\end{lem}

\begin{proof}
Based on Theorem~\ref{thm:same}, if all nodes have an OPERA chain with the same partial shape, a root in OPERA chain selects the same consensus time candidate by the Reselection algorithm.
\end{proof}

\begin{thm}
\label{thm:ct}
Lachesis consensus algorithm guarantees to reach agreement for the consensus time.
\end{thm}

\begin{proof}
For any root set $R$ in the frame $f_{i}$, time consensus algorithm checks whether more than 2n/3 roots in the frame $f_{i-1}$ selects the same value. However, each node selects one of the values collected from the root set in the previous frame by the time consensus algorithm and Reselection process. Based on the Reselection process, the time consensus algorithm can reach agreement. However, there is a possibility that consensus time candidate does not reach agreement~\cite{Fischer85}. To solve this problem, time consensus algorithm includes minimal selection frame per next $h$ frame. In minimal value selection algorithm, each root selects minimum value among values collected from previous root set. Thus, the consensus time reaches consensus by time consensus algorithm.
\end{proof}

\begin{thm}
\label{thm:bft}
If the number of reliable nodes is more than $2n/3$, event blocks created by reliable nodes must be assigned to consensus order.
\end{thm}

\begin{proof}
In OPERA chain, since reliable nodes try to create event blocks by communicating with every other nodes continuously, reliable nodes will share the event block $x$ with each other. Based on Proposition~\ref{prop:seen}, if a root $y$ in the frame $f_{i}$ Happened-Before event block $x$ and more than 2n/3 roots in the frame $f_{i+1}$ Happened-Before the root $y$, the root $y$ will be nominated as Clotho and Atropos. Thus, event block $x$ and root $y$ will be assigned consensus time $t$. 

For an event block, assigning consensus time means that the validated event block is shared by more than 2n/3 nodes. Therefore, malicious node cannot try to attack after the event blocks are assigned consensus time. When the event block $x$ has consensus time $t$, it cannot occur to discover new event blocks with earlier consensus time than $t$.
There are two conditions to be assigned consensus time earlier than $t$ for new event blocks. First, a root $r$ in the frame $f_{i}$ should be able to share new event blocks. Second, the more than 2n/3 roots in the frame $f_{i+1}$ should be able to share $r$. Even if the first condition is satisfied by malicious nodes (e.g., parasite chain),
the second condition cannot be satisfied since at least 2n/3 roots in the frame $f_{i+1}$ are already created and cannot be changed. Therefore, after an event block is validated, new event blocks should not be participate earlier consensus time to OPERA chain. 
\end{proof}
% lemma and theorem end

\subsection{Response to Attacks}\label{se:ra}
Like all other decentralized blockchain technologies, OPERA chain will likely be subject to attacks by attackers which aim to gain financial profit to damage the system. Here we describe several possible attack scenarios and how the OPERA chain intends to take preventive measures.

\subsubsection{Transaction Flooding}
A malicious participant may run a large number of valid transactions from their account under their control with the purpose of overloading the network. In order to prevent such a case, the chain intends to impose a minimal transaction fee. Since there is a transaction fee, the malicious user cannot continue to perform such attacks. Participants who participate in nodes are rewarded, and those who contribute to the ecosystem, such as by running transactions, are continuously rewarded. Such rewards are expected to be adequate in running transactions for appropriate purposes. However, since it would require tremendous cost to perform abnormal attacks, it would be difficult for a malicious attacker to create transaction flooding.

\subsubsection{Parasite chain attack}
In a DAG-based protocol, a parasite chain can be made with a malicious purpose, attempting connection by making it look like a legitimate event block. When the Main Chain is created, verification for each event block is performed. In the verification process, any event block that is not connected to the Main Chain is deemed to be invalid and is ignored, as in the case of double spending.

We suppose that less than one-third of nodes are malicious. The malicious nodes create a parasite chain. By the root definition, roots are nominated by 2n/3 node awareness. A parasite chain is only shared with malicious nodes that are less than one-third of participating nodes. A parasite chain is unable to generate roots and have a shared consensus time.

\subsubsection{Double Spending}
A double spend attack is when a malicious entity attempts to spend their funds twice. Entity $A$ has 10 tokens, they send 10 tokens to $B$ via node $n_A$ and 10 tokens to $C$ via node $n_Z$. Both node $n_A$ and node $n_Z$ agree that the transaction is valid, since $A$ has the funds to send to $B$ (according to $n_A$) and $C$ (according to $n_Z$).

Consensus is a mechanism whereby multiple distributed parties can reach agreement on the order and state of a sequence of events. Let’s consider the following 3 transactions:

- Transaction $tx_A$: $A$ (starting balance of 10) transfers 10 to $B$

- Transaction $tx_B$: $B$ (starting balance of 0) transfers 10 to $C$

- Transaction $tx_C$: $C$ (starting balance of 0) transfers 10 to $D$

We consider Node $n_A$ received the order $tx_A$ $tx_B$ $tx_C$.

The state of Node $n_A$ is $A:0$, $B:0$, $C:0$, $D:10$

Now, we consider Node $n_B$ that receives the order $tx_C$ $tx_B$ $tx_A$.

The state of Node $n_B$ is $A:0$, $B:10$, $C:0$, $D:0$

Consensus ordering gives us a sequence of events.

If the pair of event blocks $(x, y)$ has a double spending transaction, the chain can structurally detect the double spend and delay action for the event blocks until the event blocks assign time ordering.

Suppose that the pair of event blocks $(x, y)$ has same frame $f_1$. Then, all nodes must detect two event blocks before frame $f$+$2$. By the root definition, each root happened-before more than $2n/3$ previous roots. For this reason, when two roots in $f$+$1$ are selected, they must have happened-before the roots which are more than one-thirds of roots in $f$. This means that more than $2n/3$ roots in $f$+$1$ share both two roots which include the pair respectively. With the root definition and previous explanation, all roots in $f$+$2$ share both the pairs. Thus, all nodes detect the double spending event blocks at $f$+$2$ or earlier.

\subsubsection{Long-range attack}
In blockchains an adversary can create another chain. If this chain is longer than the original, the network will accept the longer chain. This mechanism exists to identify which chain has had more work (or stake) involved in its creation.

$2n/3$ participating nodes are required to create a new chain. To accomplish a long-range attack you would first need to create $>$ $2n/3$ participating malicious nodes to create the new chain.

\subsubsection{Bribery attack}
An adversary could bribe nodes to validate conflicting transactions. Since $2n/3$ participating nodes are required, this would require the adversary to bribe $>$ $n/3$ of all nodes to begin a bribery attack.

\subsubsection{Denial of Service}
LCA is a leaderless system requiring $2n/3$ participation. An adversary would have to deny $>$ $n/3$ participants to be able to successfully mount a DDoS attack.

\subsubsection{Sybil}
Each participating node must stake a minimum amount of FTM to participate in the network. Being able to stake $2n/3$ total stake would be prohibitively expensive.

\clearpage
\section{Reference}\label{se:ref}

\renewcommand\refname{\vskip -1cm}
\bibliographystyle{unsrt}
\bibliography{LCA}

\end{document}